\documentclass[fleqn,10pt]{template/wlscirep}
\usepackage[utf8]{inputenc}
\usepackage[T1]{fontenc}
\usepackage{upgreek}

\usepackage{tikz-cd}
\usepackage{tikz}
\usepackage{subcaption}

\definecolor{dark_blue}{RGB}{46,87,144}
\newcommand\rev[1]{#1}

\newtheorem{theorem}{Theorem}

\newtheorem{proof}{Proof}
\usepackage{verbatim}

\definecolor{codegreen}{rgb}{0,0.6,0}
\definecolor{codegray}{rgb}{0.5,0.5,0.5}
\definecolor{codepurple}{rgb}{0.58,0,0.82}
\definecolor{backcolour}{rgb}{0.95,0.95,0.92}

\definecolor{mypink}{rgb}{0.9254901960784314, 0.0, 0.32941176470588235}
\definecolor{myorange}{rgb}{0.996078431372549, 0.25098039215686274, 0.07058823529411765}
\definecolor{mygreen}{rgb}{0.14901960784313725, 0.6823529411764706, 0.2549019607843137}
\definecolor{myblue}{rgb}{0.17254901960784313, 0.37254901960784315, 0.6941176470588235}

\usepackage{listings}

\lstdefinestyle{mystyle}{
    backgroundcolor=\color{backcolour},   
    commentstyle=\color{mygreen},
    keywordstyle=\color{black},
    numberstyle=\tiny\color{codegray},
    stringstyle=\color{codegray},
    basicstyle=\ttfamily\footnotesize,
    breakatwhitespace=false,         
    breaklines=true,                 
    captionpos=b,                    
    keepspaces=true,                 
    numbers=left,                    
    numbersep=5pt,                  
    showspaces=false,                
    showstringspaces=false,
    showtabs=false,                  
    tabsize=2
}

\usepackage{multirow}
\usepackage{subcaption}

\newcommand{\sech}{\operatorname{sech}}
\newcommand{\argmin}{\operatorname{argmin}}
\newcommand{\nc}{\newcommand}

\usepackage{array}

\usepackage{algorithm}
\usepackage{algorithmic}
\nc{\IC}{\mathbb{C}}
\nc{\ID}{\mathbb{D}}
\nc{\IE}{\mathbb{E}}
\nc{\IN}{\mathbb{N}}
\nc{\IR}{\mathbb{R}}

\nc{\hN}{\hat{N}}

\newcommand\cls{\underset{\clap{\scriptsize cl}}{\subset}}

\nc{\be}{\begin{equation}}
\nc{\ee}{\end{equation}}

\nc{\rank}{\texttt{rank}}
\nc{\size}{\texttt{size}}
\nc{\seed}{\texttt{seed}}

\nc{\mN}{\mathcal{N}}
\nc{\eps}{\varepsilon}
\nc{\mB}{\mathcal{B}}
\nc{\mD}{\mathcal{D}}
\nc{\mO}{\mathcal{O}}
\nc{\mT}{\mathcal{T}}
\nc{\mL}{\mathcal{L}}
\nc{\mR}{\mathcal{R}}
\nc{\mE}{\mathcal{E}}

\nc{\tp}{{\tau_\textup{p}}}
\nc{\td}{{\tau_\textup{d}}}

\nc{\fA}{\mathsf{A}}
\nc{\fD}{\mathsf{D}}
\nc{\fL}{\mathsf{L}}
\nc{\mNN}{\mN\!\!\mN}

\nc{\od}{\overline{d}}
\nc{\Ntr}{N_\textrm{train}}
\nc{\Cpde}{C_\textup{pde}}
\nc{\Cquad}{C_\textup{quad}}

\nc{\bx}{{\bf x}}
\nc{\bn}{{\bf n}}
\nc{\bq}{{\bf q}}
\nc{\bW}{{\bf W}}
\nc{\bb}{{\bf b}}
\nc{\by}{{\bf y}}
\nc{\bz}{{\bf z}}
\nc{\bu}{{\bf u}}
\nc{\btheta}{{\boldsymbol{\theta}}}
\nc{\blambda}{{\boldsymbol{\lambda}}}

\usepackage{tikz-cd}
\usetikzlibrary{shapes,arrows,positioning}

\usepackage{tikz}
\usetikzlibrary{shapes,arrows,positioning}
\usetikzlibrary{babel}
\usetikzlibrary{shapes,arrows}
\definecolor{dark_blue_pers}{RGB}{46,87,144}
\definecolor{blue_pers}{RGB}{54,104,171}
\definecolor{grey_pers}{RGB}{245,245,245}
\definecolor{red_pers}{RGB}{213,78,33}
\usetikzlibrary{backgrounds, scopes}
    
\title{h-analysis and data-parallel physics-informed neural networks}
\author[1,2,*]{Paul Escapil-Inchauspé}
\author[1,2,3]{Gonzalo A. Ruz}
\affil[1]{Facultad de Ingenier\'ia y Ciencias, Universidad Adolfo Ib\'a\~nez, Santiago, Chile}
\affil[2]{Data Observatory Foundation, Santiago, Chile}
\affil[3]{Center of Applied Ecology and Sustainability (CAPES), Santiago, Chile}

\affil[*]{paul.escapil@edu.uai.cl}

\keywords{physics-informed machine learning, physics-informed neural networks, GPU acceleration, Horovod.}

\begin{abstract}We explore the data-parallel acceleration of physics-informed machine learning (PIML) schemes, with a focus on physics-informed neural networks (PINNs) for multiple graphics processing units (GPUs) architectures. In order to develop scale-robust and high-throughput PIML models for sophisticated applications which may require a large number of training points (e.g., involving complex and high-dimensional domains, non-linear operators or multi-physics), we detail a novel protocol based on $h$-analysis and data-parallel acceleration through the Horovod training framework. The protocol is backed by new convergence bounds for the generalization error \rev{and the train-test gap}. We show that the acceleration is straightforward to implement, does not compromise training, and proves to be highly efficient \rev{and controllable}, paving the way towards generic scale-robust PIML. Extensive numerical experiments with increasing complexity illustrate its robustness and consistency, offering a wide range of possibilities for real-world simulations.
\end{abstract}
\begin{document}

\flushbottom
\maketitle
%
%
\thispagestyle{empty}

\section*{Introduction}\label{sec:Introduction}
\label{sec:intro}Simulating physics throughout accurate surrogates is a hard task for engineers and computer scientists. Numerical methods such as finite element methods, finite difference methods and spectral methods can be used to approximate the solution of partial differential equations (PDEs) by representing them as a finite-dimensional function space, delivering an approximation to the desired solution or mapping\cite{steinbach2007numerical}.

Real-world applications often incorporate partial information of physics and observations, which can be noisy. This hints at using data-driven solutions throughout machine learning (ML) techniques. In particular, deep learning (DL)\cite{lecun2015deep,bengio2017deep} principles have been praised for their good performance, granted by the capability of deep neural networks (DNNs) to approximate high-dimensional and non-linear mappings, and offer great generalization with large datasets. Furthermore, the exponential growth of GPUs capabilities has made it possible to implement even larger DL models. 

Recently, a novel paradigm called physics-informed machine learning (PIML)\cite{karniadakis2021physics} was introduced to bridge the gap between data-driven\cite{YOU2021113553} and physics-based\cite{SUN2020112732} frameworks. PIML enhances the capability and generalization power of ML by adding prior information on physical laws to the scheme by restricting the output space (e.g., via additional constraints or a regularization term). This simple yet general approach was applied successfully to a wide range complex real-world applications, including structural mechanics\cite{lai2022neural,LAI2021116196} and biological, biomedical and behavioral sciences\cite{alber2019integrating}.

In particular, physics-informed neural networks (PINNs)\cite{RAISSI2019686} consist in applying PIML by means of DNNs. They encode the physics in the loss function and rely on automatic differentiation (AD)\cite{baydin2018automatic}. PINNs have been used to solve inverse problems\cite{Chen2020}, stochastic PDEs\cite{chen2021learning,MENG2020109020}, complex applications such as the Boltzmann transport equation\cite{li2022physicsboltzmann} and large-eddy simulations \cite{largeEddy}, and to perform uncertainty quantification\cite{ZHANG2019108850,escapilruzUQ2022}.

Concerning the challenges faced by the PINNs community, efficient training\cite{Wang2020WhenAW}, proper hyper-parameters setting\cite{escapil2022hyper}, and scaling PINNs\cite{shukla_xu_trask_karniadakis_2022} are of particular interest. Regarding the latter, two research areas are gaining attention.

First, it is important to understand how PINNs behave for an increasing number of training points $N$ (or equivalently, for a \rev{suitable bounded and fixed domain}, a decreasing maximum distance between points $h$). Throughout this work, we refer to this study as $h$-analysis \rev{as being the analysis of the number of training data needed to obtain a stable generalization error}. In their pioneer works\cite{MolinaroPDEs,MolinaroInverse}, Mishra and Molinaro provided a bound for the generalization error with respect to to $N$ for data-free and unique continuation problems, respectively. More precise bounds have been obtained using characterizations of the DNN\cite{shin2020convergence}.

Second, PINNs are typically trained over graphics processing units (GPUs), which have limited memory capabilities. To ensure models scale well with increasingly complex settings, two paradigms emerge: data-parallel and model-parallel acceleration. The former splits the training data over different workers, while the latter distributes the model weights. 
However, general DL backends do not readily support multiple GPU acceleration. To address this issue, Horovod\cite{khoo_lu_ying_2021} is a distributed framework specifically designed for DL, featuring a ring-allreduce algorithm\cite{sergeev2018horovod} and implementations for TensorFlow, Keras and PyTorch.

As model size becomes prohibitive, domain decomposition-based approaches allow for distributing the computational domain. Examples of such approaches include conservative PINNs (cPINNs) \cite{JAGTAP2020113028}, extended PINNs (XPINNs) \cite{CiCP-28-2002,hu2021extended}, and distributed PINNs (DPINNs) \cite{dwivedi2019distributed}. cPINNs and XPINNs were compared in \cite{shukla2021parallel}. These approaches are compatible with data-parallel acceleration within each subdomain. Additionally, a recent review concerning distributed PIML\cite{shukla_xu_trask_karniadakis_2022} is also available. Regarding existing data-parallel implementations, TensorFlow MirroredStrategy in TensorDiffEq\cite{mcclenny2021tensordiffeq} and \href{https://developer.nvidia.com/modulus}{NVIDIA Modulus}\cite{hennigh2021nvidia}, should be mentioned. However, to the authors knowledge, there is no systematic study of the background of data-parallel PINNs and their implementation.

In this work, we present a procedure to attain data-parallel efficient PINNs. It relies on $h$-analysis and \rev{is backed by} a Horovod-based acceleration. Concerning $h$-analysis, we observe PINNs exhibiting three phases of behavior as a function of the number of training points $N$:
\begin{enumerate}
\item A pre-asymptotic regime, where the model does not learn the solution due to missing information;
\item A transition regime, where the error decreases with $N$;
\item A permanent regime, where the error remains stable.
\end{enumerate}
To illustrate this, Fig.~\ref{fig:transient} presents the relative $L^2$ error distribution with respect to $N_f$ (number of domain collocation points) for the forward ``\hyperref[{1DLaplace}]{1D Laplace}'' case. The experiment was conducted over $8$ independent runs with a learning rate of $10^{-4}$ and $20000$ iterations of ADAM\cite{kingma2014adam} algorithm. The transition regime---where variability in the results is high and some models converge while others do not---is between $N_f=\rev{64}$ and $N_f=\rev{400}$. For more information on the experimental setting and the definition of precision $\rho$, please refer to ``\hyperref[{1DLaplace}]{1D Laplace}''.

\begin{figure}[!htb]
\center
\includegraphics[width=\textwidth]{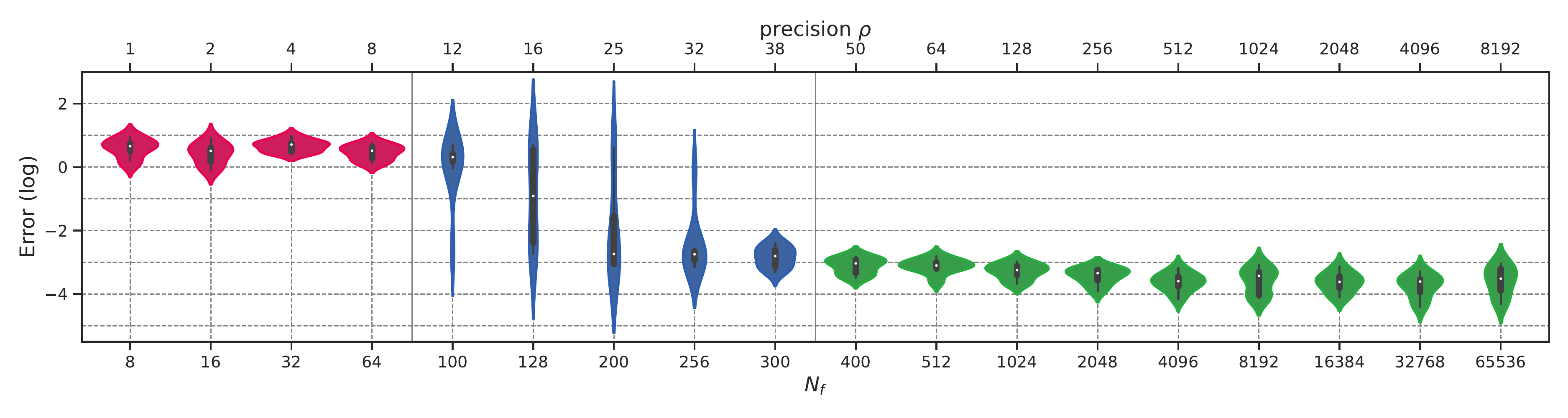}
\caption{Error v/s the number of domain collocation points $N_f$ for the ``\hyperref[{1DLaplace}]{1D Laplace}'' case. A pre-asymptotic regime (pink) is followed by a rapid transition regime (blue), and eventually leading to a permanent regime (green). This transition occurs over a few extra training points.} 
\label{fig:transient}
\end{figure}

Building on the empirical observations, we use the setting in \cite{MolinaroPDEs,MolinaroInverse} to supply a rigorous theoretical background to $h$-analysis. One of the main contributions of this manuscript is the bound on the ``\hyperref[{thm:coupledgeneralization}]{Generalization error for generic PINNs}'', which allows for a simple analysis of the $h$-dependence. \rev{Furthermore, this bound is accompanied by a practical ``\hyperref[{thm:traintest}]{Train-test gap bound}'', supporting regimes detection.}

\rev{To summarize the latter results, a simple yet powerful recipe for any PIML scheme could be: 
\begin{enumerate}
    \item Choose the right model and hyper-parameters to achieve a low training loss;
    \item Use enough training points $N$ to reach the permanent regime (e.g., such that the training and test losses are similar).
\end{enumerate}}

Any practitioner strives to reach the permanent regime for their PIML scheme, and we provide the necessary details for an easy implementation of Horovod-based data acceleration for PINNs, with direct application to any PIML model. Fig.~\ref{fig:DataParallel} (left) further illustrates the scope of data-parallel PIML. \rev{For the sake of clarity, Fig.~\ref{fig:DataParallel} (right) supplies a comprehensive review of important notations defined throughout this manuscript, along with their corresponding introductions.}
\begin{figure}[!htb]
\centering
 \begin{minipage}{.49\textwidth}
\centering
  \includegraphics[width=\textwidth]{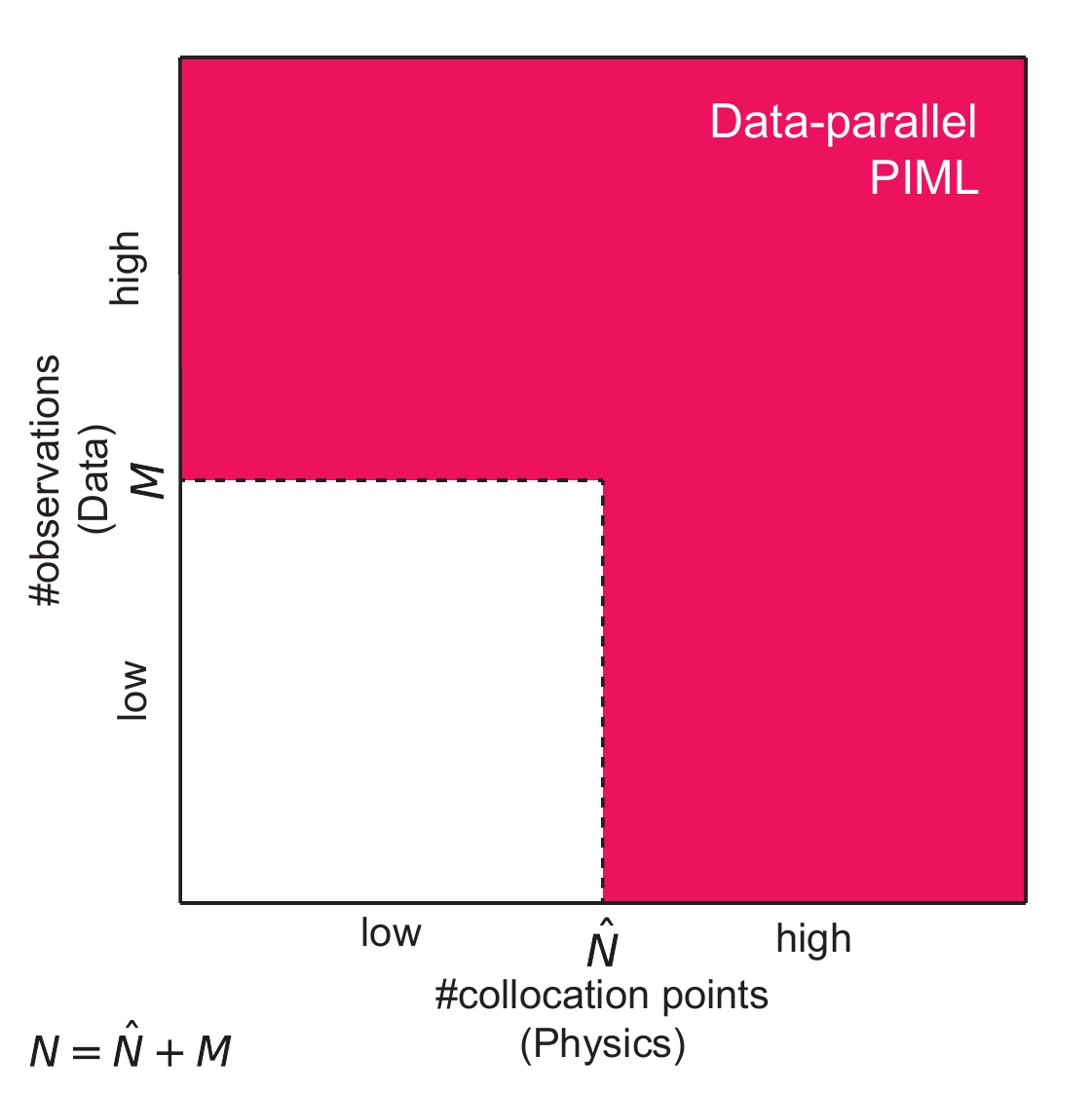}
\end{minipage}
 \begin{minipage}{.49\textwidth}
  \centering
  \vspace{-1cm}
 \begin{table}[H]
\renewcommand\arraystretch{2}
\begin{center}
\footnotesize
\begin{tabular}{
    |>{\centering\arraybackslash}m{1.2cm}
    |>{\centering\arraybackslash}m{2.4cm}
    |>{\centering\arraybackslash}m{1cm}|
    >{\centering\arraybackslash}m{3cm}
    |>{\centering\arraybackslash}m{2cm}|
    }
    \hline   
Notation & Definition & Eq. \\ \hline\hline
$\Lambda$ & $  = \{f,g,\hbar,u\}$ & Eq.~\eqref{eq:lambda} \\ \hline 
$\xi_\nu $ & Residual for $\nu \in \Lambda$ & Eq.~\eqref{eq:PDEresiduals} \\ \hline \hline 
$\hat{N}$ & \# collocation points  &  Eq.~\eqref{eq:MN} \\ \hline 
$M$ & \# observations  &  Eq.~\eqref{eq:MN} \\ \hline 
$N$ & \# training points  &  Eq.~\eqref{eq:MN} \\ \hline \hline 
$\eps_{T,\cdot}$ & Training error &   Eq.~\eqref{eq:training_loss} \\ \hline 
$\eps_{V, \cdot}$ & Testing error &   Eq.~\eqref{eq:training_loss} \\ \hline 
$\eps_{G, \cdot}$ & Generalization error &   Eq.~\eqref{eq:generalizationdef} \\ \hline 
\end{tabular}
\end{center} 
\end{table} 
  \end{minipage}
   \caption{\rev{Left: Scope of of data-parallel PIML. Right: Comprehensive review of important notations defined throughout this manuscript, along with their corresponding introductions.}}
   \label{fig:DataParallel}
\end{figure}

Next, we apply the procedure to increasingly complex problems and demonstrate that Horovod acceleration is straightforward, using the pioneer \href{https://github.com/maziarraissi/PINNs}{PINNs} code of Raissi as an example. Our main practical findings concerning data-parallel PINNs for up to $8$ GPUs are the following :
\begin{itemize}\setlength{\itemsep}{0pt}
\item They do not require to modify their hyper-parameters;
\item They show similar training convergence to the $1$ GPU-case;
\item They lead to high efficiency for both weak and strong scaling (e.g~$E_\text{ff} \rev{> 80\% }$ for Navier-Stokes problem with 8 GPUs).
\end{itemize}

This work is organized as follows: In ``\hyperref[{sec:Problem}]{Problem formulation}'', we introduce the PDEs under consideration, PINNs and convergence estimates for the generalization error. We then move to ``\hyperref[{sec:hvd}]{Data-parallel PINNs}'' and present ``\hyperref[{sec:Numexp}]{Numerical experiments}''. Finally, we close this manuscript in ``\hyperref[{sec:Conclusion}]{Conclusion}''.

\section*{Problem formulation}\label{sec:Problem}
\subsection*{General notation}\label{subsec:Notation}Throughout, vector and matrices are expressed using bold symbols. For a natural number $k$, we set $\IN_k:= \{k,k+1,\cdots\}$. For $p \in \IN_0 = \{0,1,\cdots,\}$, and an open set $D\subseteq \IR^d$ with $d\in \IN_1$, let $L^p(D)$ be the standard class of functions with bounded $L^p$-norm over $D$. Given $s\in \IR^+$, we refer to\cite[Section 2]{steinbach2007numerical} for the definitions of Sobolev function spaces $H^s(D)$. Norms are denoted by $\|\cdot\|$, with subscripts indicating the associated functional spaces. For a finite set $\mT$, we introduce notation $|\mT| := \textup{card}(\mT)$, closed subspaces are denoted by a $\cls$-symbol and $\imath^2=-1$.

\subsection*{Abstract PDE}\label{subsec:AbstractProblem}In this work, we consider a domain $D \subset \IR^d$, $d \in \IN_1$, with boundary $\Gamma = \partial D$. For any $T>0$, $\ID := D\times [0,T]$, we solve a general non-linear PDE of the form:
\be\label{eq:PDE}
\begin{cases}
\mN[u(\bx,t);\blambda] &= f(\bx,t) \quad  (\bx,t) \in  D_f =: D \times [0,T],\\
\mB[u(\bx,t); \blambda] &= g(\bx,t), \quad (\bx ,t) \in D_g =: \Gamma \times [0,T],\\
u (\bx,0)  & = \hbar(\bx), \quad  \bx \in D_{\hbar} =: D,
\end{cases}
\ee 
with $\mN$ a spatio-temporal differential operator, $\mB$ the boundary conditions (BCs) operator, $\blambda$ the material parameters---the latter being unknown for inverse problems---and $u(\bx,t)\in \IR^m$ for any $m\in \IN_1$. Accordingly, for any function $\hat{u}$ defined over $\ID$, we introduce
\be \label{eq:lambda}
\rev{\Lambda :=\{f,g,\hbar,u\}}
\ee 
and define the residuals $\xi_v$ for each $v \in \Lambda$ and any observation function $u_\textup{obs}$:
\be\label{eq:PDEresiduals}
\begin{cases}
\xi_f(\bx,t;\blambda)& : = \mN[\hat{u}(\bx,t);\blambda] - f(\bx,t) \quad\textup{in} \quad  D_f,\\
\xi_g (\bx,t;\blambda)& : = \mB[\hat{u}(\bx,t);\blambda]- g(\bx,t) \quad \textup{in} \quad D_g,\\
\xi_\hbar (\bx,0) & : =\hat{u} (\bx,0)  - \hbar(\bx) \quad   \textup{in} \quad D_\hbar, \\
\xi_{u} (\bx,t) & : = \hat{u}(\bx,t) - u_\textup{obs}(\bx,t)\quad   \textup{in} \quad D_u := \ID.
\end{cases}
\ee 

\subsection*{PINNs}Following\cite{escapil2022hyper,lu2021deepxde}, let $\sigma $ be a smooth activation function. Given an input $(\bx,t) \in \IR^{d+1}$, we define $\mNN_\theta$ as being a $L$-layer neural feed-forward neural network with $W_0= d+ 1$, $W_L = m$ and $W_l$ neurons in the $l$-th layer for $1 \leq l \leq L-1$. For constant width DNNs, we set $W=W_1=\cdots = W_{L-1}$. For $1 \leq l \leq L$, let us denote the weight matrix and bias vector in the $l$-th layer by $\bW^l \in \IR^{d_l \times d_{l-1}}$ and $\bb^l \in \IR^{d_l}$, respectively, resulting in:
\be \label{eq:DNN}
\begin{array}{rll}
\text{input layer:} \quad & (\bx,t)\in \IR^{d+1},\\
\text{hidden layers:} \quad & \bz^l (\bx) = \sigma ( \bW^l \bz^{l-1} (\bx) + \bb^l) \in \IR^{d_l} & \quad \text{for} \quad 1 \leq l \leq L-1,\\
\text{output layer:} \quad & \bz^L(\bx) = \bW^L \bz^{L-1} (\bx) + \bb^L  \in \IR^m. 
\end{array}
\ee
This results in representation $\bz^L(\bx,t)$, with 
\be\label{eq:weights}\theta:= \left\{ (\bW^1, \bb^1), \cdots, (\bW^L, \bb^L)\right\},
\ee 
the (trainable) parameters---or weights---in the network. We set $\Theta = \IR^{|\Theta|}$. Application of PINNs to Eq.~\eqref{eq:PDE} yields the approximate $u_\theta(\bx,t) = \bz^L(\bx,t)$. 

We introduce the training dataset $\mT_v:=\{\tau_v^i\}_{i=1}^{N_v}$, $\tau_v^i \in D_v$, $N_v \in \IN$ for $i = 1,\cdots, N_v, v \in \Lambda$ and observations $u_\text{obs}(\tau_u^i)$, $i=1,\cdots,N_u$. Furthermore, to each training point $\tau_v^i$ we associate a quadrature weight $w_v^i>0$. All throughout this manuscript, we set:
\be\label{eq:MN}
M:= N_u, \quad \hN := N_f + N_g + N_\hbar \quad \text{and} \quad  N := \hN + M.
\ee
Note that $M$ (resp.~$\hN$) represents the amount of information for the data-driven (resp.~physics) part, by virtue of the PIML paradigm (refer to Fig.~\ref{fig:DataParallel}). The network weights $\theta$ in Eq.~\eqref{eq:weights} are trained (e.g.,~via ADAM optimizer\cite{kingma2014adam}) by minimizing the weighted loss:
\be
\label{eq:loss}
\mL_\theta: = \sum_{v\in \Lambda} \upomega_v \mL_\theta^v,\quad\textup{wherein}\quad \mL_\theta^v :=\sum_{i=1}^{N_v} w_v^i|\xi_{v,\theta}(\tau^i_v)|^2 \quad \textup{and} \quad \upomega_v >0 \quad \textup{for} \quad v \in \Lambda.
\ee 
We seek at obtaining:
\be\label{eq:minimizer}
\theta^\star := \argmin_{\theta \in \Theta } (\mL_\theta).
\ee 
The formulation for PINNs addresses the cases with no data (i.e.~$M=0$) or physics (i.e.~$\hN=0$), thus exemplifying the PIML paradigm. Furthermore, it is able to handle time-independent operators with only minor changes; a schematic representation of a forward time-independent PINN is shown in Fig.~\ref{picture:PINNs}.
\begin{figure}[!htb]
\centering
\includegraphics[width=.65\textwidth]{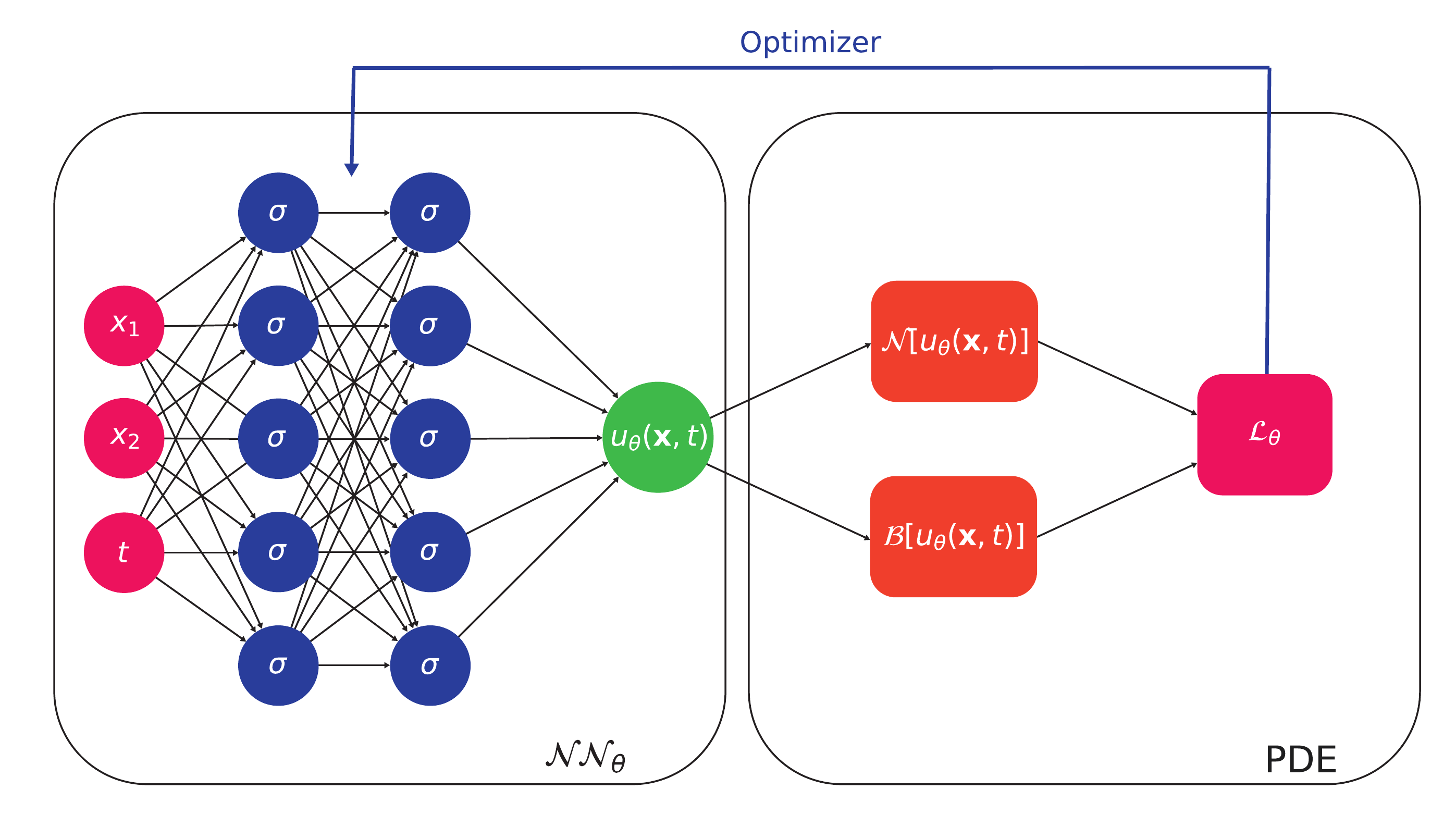}
\caption{Schematic representation of a PINN. A DNN with $L=3$ (i.e. $L-1 = 2$ hidden layers) and $W=5$ learns the mapping $\bx \mapsto u(\bx\rev{,t})$. The PDE is taken into account throughout the residual $\mL_\theta$, and the trainable weights are optimized, leading to optimal $\theta^\star$.}
\label{picture:PINNs}
\end{figure}
Our setting assumes that the material parameters $\blambda$ are known. If $M>0$, one can solve the inverse problem by seeking:
\be 
(\theta^\star_\text{inverse},\blambda^\star_\text{inverse}) : = \argmin_{\theta \in \Theta, \blambda} \mL_\theta[\blambda].
\ee
Similarly, unique continuation problems\cite{MolinaroInverse}, which assume incomplete information for $f,g$ and $\hbar$, are solved throughout PINNs without changes. Indeed, ``\hyperref[{NavierStokes}]{2D Navier-Stokes}'' combines unique continuation problem and unknown parameters $\lambda_1,\lambda_2$.

\subsection*{Automatic Differentiation}We aim at giving further details about back-propagation algorithms and their dual role in the context of PINNs:
\begin{enumerate}
\item Training the DNN by calculating $\frac{\partial \mL_\theta}{\partial \theta}$;
\item Evaluating the partial derivatives in $\mN [u_\theta(\bx,t);\blambda]$ and $\mB[u_\theta(\bx,t);\blambda]$ so as to compute the loss $\mL_\theta$.
\end{enumerate}
They consist in a forward pass to evaluate the output $u_\theta$ (and $\mL_\theta$), and a backward pass to assess the derivatives. To further elucidate back-propagation, we reproduce the informative diagram from\cite{baydin2018automatic} in Fig.~\ref{fig:AD}.
\begin{figure}[!htb]
\center
\includegraphics[width=.65\textwidth]{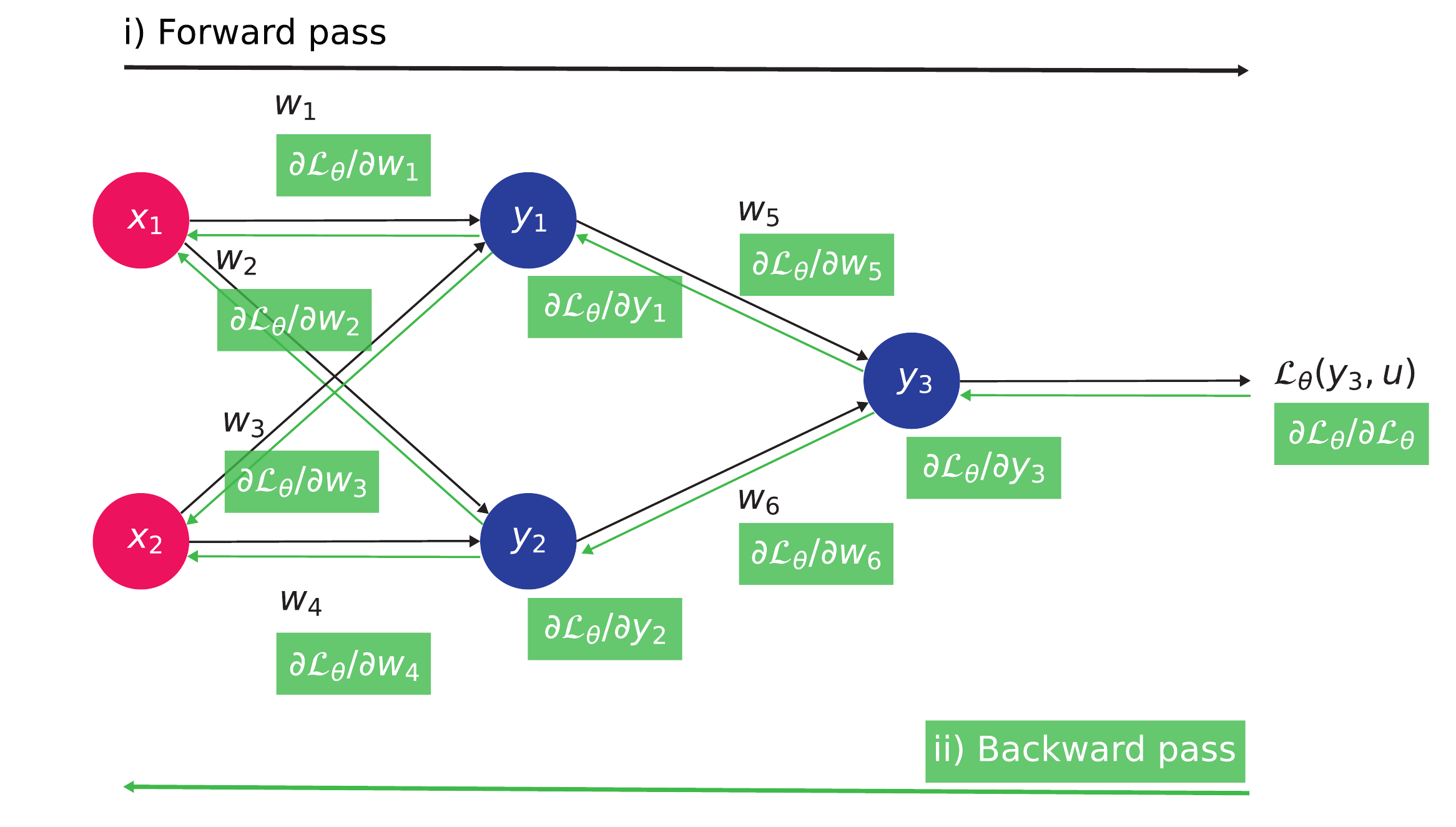}
\caption{Overview of back-propagation. A forward pass generates activations $y_i$ and computes the error $\mL_\theta(y_3,u)$. This is followed by a backward pass, through which the error adjoint is propagated to obtain the gradient with respect to weights $\nabla \mL_\theta$ where $\theta=(\mathrm{w}_1,\cdots,\mathrm{w}_6)$. Additionally, \rev{spatio-temporal partial derivatives} can be computed in the same backward pass.}
\label{fig:AD}
\end{figure}
TensorFlow includes \href{https://www.tensorflow.org/guide/autodiff}{reverse mode AD} by default. Its cost is bounded with $|\Theta|$ for scalar output NNs (i.e.~for $ m=1$). The application of back-propagation (and reverse mode AD in particular) to any training point is independent of other information, such as neighboring points or the volume of training data. This allows for data-parallel PINNs. Before detailing its implementation, we justify the $h$-analysis through an abstract theoretical background.

\subsection*{Convergence estimates}To understand better how PINNs scale with $N$, we follow the method in \cite{MolinaroPDEs,MolinaroInverse} under a simple setting, allowing to control the data and physics counterparts in PIML. Set $s \geq 0$ and define spaces: 
\be 
\hat{Y} \cls Y^\star  \cls Y = L^2(\ID,\IR^m)  \quad \text{and} \quad \hat{X}  \cls X^\star  \cls X = H^s(\ID,\IR^m)
\ee
We assume that Eq.~\eqref{eq:PDE} can be recast as:
\begin{align}
\label{eq:abstract}
\fA  u&  = b \quad \textup{with} \quad \fA : X^\star \to Y^\star \quad \textup{and} \quad b \in Y^\star,\\
u&  = u_\text{obs} \quad \textup{in}\quad  X^\star.
\end{align}
We suppose that Eq.~\eqref{eq:abstract} is well-posed and that for any $u,v \in \hat{X}$, there holds that:
\be\label{eq:stability}
\|u - v\|_{Y} \leq  \Cpde (\|u\|_{\hat{X}}, \|v\|_{\hat{X}}) \left( \|\fA u- \fA v\|_Y\right). 
\ee
Eq.~\eqref{eq:abstract} is a stability estimate, allowing to control the total error by means of a bound on PINNs residual. Residuals in Eq.~\eqref{eq:PDEresiduals} are:
\be\label{eq:PDEresidualsAbstract}
\begin{cases}
\xi_D& : = \fA u - b \quad \text{in}\quad Y^\star,\\
\xi_{u} &  = u - u_\textup{obs}\quad \text{in}\quad X^\star .
\end{cases}
\ee
From the expression of residuals, we are interested in approximating integrals:
\be \label{eq:functionApproximation} 
\overline{g} = \int_\ID g(y) dy \quad \text{and} \quad  \overline{l} = \int_\ID l(z) dz \quad \text{for} \quad g \in \hat{Y}, l \in \hat{X}.
\ee
We assume that we are provided quadratures:
\be 
\overline{g}_{\hN} = \sum_{i=1}^{\hN} \rev{w^i_D} g(\rev{\tau^i_D}) \quad \text{and} \quad \overline{l}_M = \sum_{i=1}^M \rev{w^i_u} l(\rev{\tau^i_u})
\ee
for weights $\rev{w^i_D},\rev{w^i_u}$ and quadrature points $\rev{\tau_D^i}, \rev{\tau_u^i} \in \ID$ such that for $\alpha,\beta > 0$:
\be \label{eq:quad_rate}
|\overline{g} - {\overline{g}}_{\hN} |\leq C_{\textup{quad},Y} {\hN}^{-\alpha} \quad \text{and} \quad  |\overline{l} - \overline{l}_M | \leq C_{\textup{quad},X} M^{-\beta}.
\ee
For any $\upomega_u>0$, the loss is defined as follows: 
\be
\begin{split}
\label{eq:training_loss}
\mL_\theta &= \sum_{i=1}^{\hN} \rev{w^i_D} |\xi_{D,\theta} (\rev{\tau^i_D})|^2 +  \upomega_u \sum_{i=1}^M \rev{w_u^i} |\xi_{u,\theta} (\rev{\tau_u^i})|^2  \approx \|\xi_{D,\theta}\|_{Y}^2 + \upomega_u \|\xi_{u,\theta} \|_{X}^2\\
& =  \eps_{T,D}^2 + \upomega_u\eps_{T,u}^2,
\end{split}
\ee
with $\eps_{T,D}$ and $\eps_{T,u}$ the training error for collocation points and observations respectively.

Notice that application of Eq.~\eqref{eq:quad_rate} to $\xi_{D,\theta}$ and $\xi_{u,\theta}$ yields:
\be \label{eq:quadratures1}
|\|\xi_{D,\theta}\|^2_Y -   \eps_{T,D}^2 | \leq  C_{\textup{quad},Y}\hN^{-\alpha}\quad \text{and} \quad |\|\xi_{u,\theta}\|_X^2 - \eps_{T,u}^2 | \leq  C_{\textup{quad},X}M^{-\beta }.
\ee
We seek to quantify the \emph{generalization error}:
\be\label{eq:generalizationdef}
\eps_G = \eps_G(\theta^\star): = \|u - u^\star \|_X \quad \textup{with } u^\star :=  u_{\theta^\star} \quad \text{and}\quad \theta^\star := \argmin_\theta \mL_\theta .
\ee
We detail a new result concerning the generalization error for PINNs.
\begin{theorem}[Generalization error for generic PINNs]\label{thm:coupledgeneralization}
Under the presented setting, there holds that:
\be \label{eq:theoremgeneralization}
\eps_G \leq \frac{\Cpde}{1+\upomega_u}  \left(\eps_{T,D} + C_{\textup{quad},Y}^{1/2} \hN^{-\alpha/2}\right) + \frac{\upomega_u}{1+\upomega_u}  \left( \eps_{T,u} +  C_{\textup{quad},X}^{1/2} M^{-\beta/2}  + \hat{\mu} \right)
\ee
with $\hat{\mu}: = \|u-u_\textup{obs}\|_X$.
\end{theorem}
\begin{proof}Consider the setting of Theorem~\ref{thm:coupledgeneralization}. There holds that:
\begin{align*}
(1 + \upomega_u)\eps_G  &= \|u - u^\star \|_X + \upomega_u \|u - u^\star \|_X, \quad  \textup{by Eq.~\eqref{eq:generalizationdef}} \\
& \leq \Cpde \|\fA u - \fA u^\star\|_Y + \upomega_u  \|u-u^\star\|_X, \quad  \textup{by Eq.~\eqref{eq:stability}}\\
& \leq \Cpde \|\xi_{D,\theta^\star}\|_Y +\upomega_u \|u - u_\textup{obs}\|_X  + \upomega_u \|u_\textup{obs}- u^\star\|_X , \quad  \textup{by Eq.~\eqref{eq:PDEresidualsAbstract} and triangular inequality}\\ 
& = \Cpde \|\xi_{D,\theta^\star}\|_Y + \upomega_{u} \|\xi_{u,\theta^\star}\|_Y + \upomega_u\hat{\mu}, \quad \textup{by definition}\\
& \leq \Cpde\eps_{T,D} + \upomega_u \eps_{T,u} +  \Cpde C_{\textup{quad},Y}^{1/2} \hN^{-\alpha/2} + \upomega_u C_{\textup{quad},X}^{1/2} M^{-\beta/2}  + \upomega_u \hat{\mu}, \quad \textup{by Eq.~\eqref{eq:quadratures1}}.
\end{align*}
\end{proof}

The novelty of Theorem \ref{thm:coupledgeneralization} is that it describes the generalization error for a simple case involving collocation points and observations. \rev{It states that the PINN generalizes well as long as the training error is low and that sufficient training points are used. }To make the result more intuitive, we rewrite Eq.~\eqref{eq:theoremgeneralization}, with $\sim$ expressing the terms up to positive constants:
\be 
\eps_G \sim (\hN^{-\alpha/2} + M^{-\beta/2}) + \eps_{T,D} + \eps_{T,u} + \hat{\mu} .
\ee
The generalization error depends on the training errors (which are tractable during training), parameters $\hN$ and $M$ and bias $\hat{\mu}$.

To return to $h$-analysis, we now have a theoretical proof of the three regimes presented in \rev{``\hyperref[{sec:Introduction}]{Introduction}''}. Let us assume that $\hat{\mu}=0$. For small values of $\hN$ or $M$, the bound in Theorem~\ref{thm:coupledgeneralization} is too high to yield a meaningful estimate. Subsequently, the convergence is as $\max(\hN^{-\alpha/2},M^{-\beta/2})$, marking the transition regime. It is paramount for practitioners to reach the permanent regime when training PINNs, giving ground to data-parallel PINNs.

\rev{In general applications, the exact solution $u$ is not available. Moreover, it is relevant to determine whether $N$ is large enough. To this extent, we introduce a same cardinality testing (or validation) set. }\rev{Interestingly, the entire analysis above and Theorem \ref{thm:coupledgeneralization} remain valid for another set of testing points, with the testing error $\eps_{V,D}$ and $\eps_{V,u}$ set as in Eq.~\eqref{eq:training_loss}. The train-test gap, which is tractable, can be quantified as follows.}
\rev{
\begin{theorem}[Train-test gap bound]\label{thm:traintest}Under the presented setting, there holds that:
$$
|\eps_{T,D} - \eps_{V,D}| \leq 2 C_{\textup{quad},Y}^{1/2} \hat{N}^{-\alpha /2} \quad \text{and} \quad |\eps_{T,u} - \eps_{V,u}| \leq 2 C_{\textup{quad},X}^{1/2} M^{-\beta /2} .
$$
\end{theorem}
}
\rev{
\begin{proof}Consider the setting of Theorem~\ref{thm:traintest}. For $v \in \{D, u\}$ and $\cdot \in \{ Y,X\}$ there holds that:
\begin{align*}
|\eps_{T,v} - \eps_{V,v}|  &\leq \left| \eps_{T,v} - \|\xi_{v,\theta^\star}\|^2_{\cdot} \right| + \left|  \eps_{V,v}- \|\xi_{v,\theta^\star}\|^2_{\cdot} \right| , \quad  \textup{by triangular inequality} \\
& \leq   2 C_{\textup{quad}, \cdot }^{1/2} N_v^{-\alpha /2}, \quad \textup{by Eq.~\eqref{eq:quadratures1}}.
\end{align*}
\end{proof}}
\rev{The bound in Theorem \ref{thm:coupledgeneralization} is valuable as it allows to assess the quadrature error convergence---and the regime---with respect to the number of training points. }
\section*{Data-parallel PINNs}\label{sec:hvd}
\subsection*{Data-distribution and Horovod}In this section, we present the data-parallel distribution for PINNs. Let us set $\size \in \IN_1$ and define ranks (or workers):
$$
\rank = 0 , \cdots, \size-1,
$$
each rank corresponding generally to a GPU. Data-parallel distribution requires the appropriate partitioning of the training points across ranks.

We introduce $\hN_1,M_1\in \IN_1$ collocation points and observations, respectively, for each $\rank$ (e.g., a GPU) yielding:
$$\mT_v = \bigcup_{\rank =0}^{\size-1} \mT_v^\rank \quad \textup{for}\quad v \in \{D,u\}, $$
with
\be 
\hN = \size \times \hN_{1}, \quad M = \size \times M_1 \quad \text{and}\quad  \mT = \mT_{\rev{D}} \cup \mT_u \quad \text{with} \quad N = \hN + M.
\ee 

Data-parallel approach is as follows: We send the same synchronized copy of the DNN $\mN\!\!\mN_\theta$ defined in Eq.~\eqref{eq:DNN} to each rank. Each rank evaluates the loss $\mL^\rank_\theta$ and the gradient $\nabla_\theta \mL^\rank_\theta$. The gradients are then averaged using an all-reduce operation, such as the ring all-reduce implemented in Horovod\cite{sergeev2018horovod,patarasuk2009bandwidth}, which is known to be \rev{bandwith} optimal with respect to the number of ranks\cite{patarasuk2009bandwidth}. The process is illustrated in Figure \ref{fig:allReduce} for $\size =4$. The ring-allreduce algorithm involves each of the $\size$ nodes communicating with two of its peers $2 \times (\size -1)$ times \cite{sergeev2018horovod}. 

\tikzset{every picture/.style={line width=0.75pt}}
\begin{figure}[!htb]
\centering
\includegraphics[width=.4\textwidth]{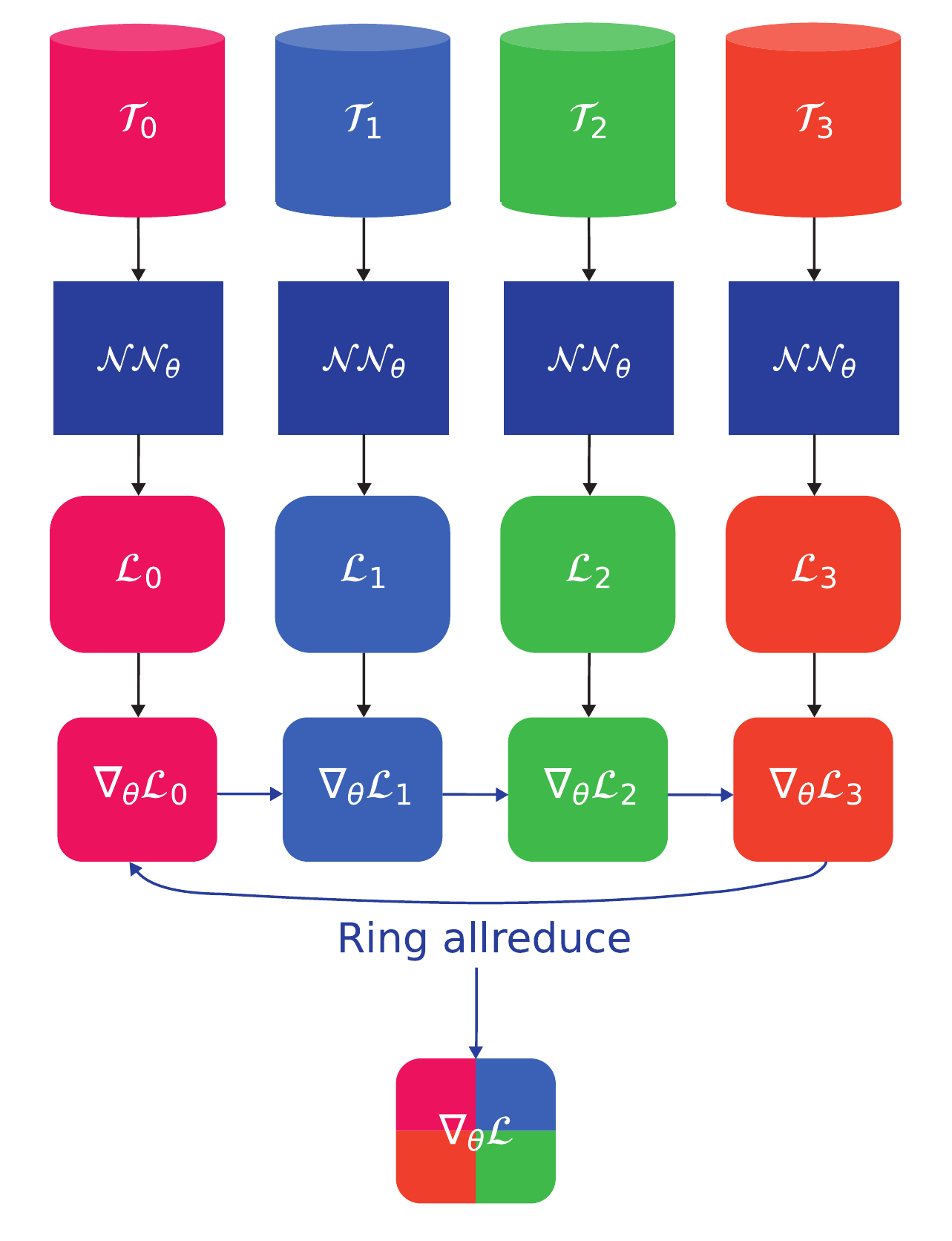}
\caption{Data-parallel framework. Horovod supports ring-allreduce algorithm.}
\label{fig:allReduce}
\end{figure}

It is noteworthy to observe that data generation for data-free PINNs (i.e.~with $M=0$) requires no modification to existing codes, provided that each rank has a different seed for random or pseudo-random sampling. Horovod allows to apply data-parallel acceleration with minimal changes to existing code. Moreover, our approach and Horovod can easily be extended to multiple computing nodes. As pointed out in ``\hyperref[{sec:Introduction}]{Introduction}'', \rev{Horovod} supports popular DL backends such as TensorFlow, PyTorch and Keras. In Listing \ref{lst:tensor},  we demonstrate how to integrate data-parallel distribution using Horovod with a generic PINNs implementation in TensorFlow 1.x.  The highlighted changes in pink show the steps for incorporating Horovod, which include: (i) initializing Horovod; (ii) pinning available GPUs to specific workers; (iii) wrapping the Horovod distributed optimizer and (iv) broadcasting initial variables to the master rank being $\rank = 0$.

{\centering
\begin{minipage}{.75\linewidth}
\begin{lstlisting}[language=Python, caption=Horovod for PINNs with TensorFlow 1.x. Data-parallel Horovod PINNs require minor changes to existing code., label=lst:tensor]
# Initialize Horovod
<@\textcolor{mypink}{import horovod.tensorflow as hvd}@>
<@\textcolor{mypink}{hvd.init()}@>

# Pin GPU to be used to process local rank (one GPU per process)
config = tf.ConfigProto()
config.gpu_options.visible_device_list = str(<@\textcolor{mypink}{hvd.local\_rank()}@>)

# Build the PINN
loss = ...
opt = tf.train.AdamOptimizer()
# Add Horovod Distributed Optimizer
<@\textcolor{mypink}{opt = hvd.DistributedOptimizer(opt)}@>

train = opt.minimize(loss)

# Initialize variables
init = tf.global_variables_initializer()
self.sess.run(init)

# Broadcast variables from rank 0 to other workers
<@\textcolor{mypink}{bcast = hvd.broadcast\_global\_variables(0)}@>
<@\textcolor{mypink}{self.sess.run(bcast)}@>

# Train the model
while n <= maxiter:
    sess.run(train)    
    n += 1
\end{lstlisting}
\end{minipage}\par 
}

\subsection*{Weak and strong scaling}Two key concepts in data distribution paradigms are weak and strong scaling, which can be explained as follows: Weak scaling involves increasing the problem size proportionally with the number of processors, while strong scaling involves keeping the problem size fixed and increasing the number of processors. To reformulate:
\begin{itemize}
\setlength{\itemsep}{0pt}
\item Weak scaling: Each worker has $(\hN_1,M_1)$ training points, and we increase the number of workers $\size$;
\item Strong scaling: We set a fixed total number of $(\hN_1,M_1)$ training points, and we split the data over increasing $\size$ workers. 
\end{itemize}
We portray weak and strong scaling in Fig.~\ref{fig:scalings} for a data-free PINN with $\hN_1=16$. Each box represents a GPU, with the number of collocation points as a color. On the left of each scaling option, we present the unaccelerated case. \rev{Finally, we} introduce the training time $t_\size$ for $\size$ workers. This allows to define the efficiency and speed-up as: 
$$
E_{\text{ff}}:= \frac{t_1}{t_\size}\quad \textup{and} \quad S_{\text{up}} := \size \frac{t_1}{t_\size}.
$$
\begin{figure}[!htb]
\center
\includegraphics[width=\textwidth]{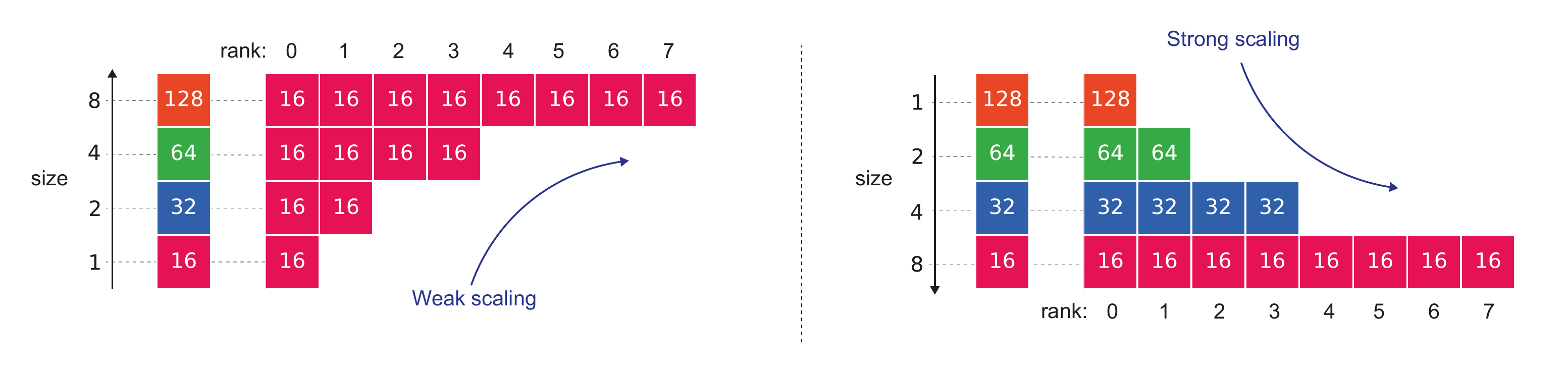}
\caption{Weak and strong scaling for $\hN_1=16$ and $\size=8$.}
\label{fig:scalings}
\end{figure}

\section*{Numerical experiments}\label{sec:Numexp}Throughout, we apply our proceeding to three cases of interest:
\begin{itemize}\setlength{\itemsep}{0pt}
\item ``\hyperref[{1DLaplace}]{1D Laplace}'' equation (forward problem);
\item ``\hyperref[{Schrodinger}]{1D Schrödinger}'' equation (forward problem);
\item ``\hyperref[{NavierStokes}]{2D Navier-Stokes}'' equation (inverse problem).
\end{itemize}
For each case, we perform a $h$-analysis followed by Horovod data-parallel acceleration, which is applied to the domain training points (and observations for the Navier-Stokes case). Boundary loss terms are negligible due to the sufficient number of boundary data points.
\subsection*{Methodology}
We perform simulations in single float precision on a AMAX DL-E48A AMD Rome EPYC server with 8 Quadro RTX 8000 Nvidia GPUs---each one with a 48 GB memory. We use a \href{https://hub.docker.com/r/pescapil/dist-training-horovod}{Docker image} of Horovod 0.26.1 with \rev{CUDA 12.1,} Python 3.6.9 and \href{https://www.tensorflow.org/}{Tensorflow} 2.6.2. All throughout, we use \texttt{tensorflow.compat.v1} as a backend without eager execution.

All the results are ready for use in \href{https://github.com/pescap/HorovodPINNs}{HorovodPINNs} GitHub repository and fully reproducible\rev{, ensuring also compliance with FAIR principles (Findability, Accessibility, Interoperability, and Reusability) for scientific data management and stewardship \cite{wilkinson2016fair}.} We run experiments $8$ times with seeds defined as:
$$
\seed + 1000 \times \rank,
$$
in order to obtain rank-varying training points. For domain points, Latin Hypercube Sampling is performed with \href{https://pythonhosted.org/pyDOE/}{pyDOE} 0.3.8. Boundary points are defined over uniform grids. 

We use Glorot uniform initialization \cite[Chapter 8]{bengio2017deep}. ``Error'' refers to the $L^2$-relative error taken over $\mT^\textup{test}$, and ``Time\rev{''} stands for the training time in seconds. For each case, the loss in Eq.~\eqref{eq:loss} is with unit weights $\upomega_v=1$ and Monte-Carlo quadrature rule $\rev{w_v^i} = \frac{1}{N_v}$ for $v\in \Lambda$. Also, we set $\text{vol}(\ID)$ the volume of domain $\ID$ and
$$
\rho : =  \frac{{N_f}^{1/ (d+1)} }{{\text{vol}}(\ID)^{1 / (d+1)} }\rev{.}
$$
We introduce $t^k$ the time to perform $k$ iterations. The training points processed by second is as follows:
\be \label{eq:pointsec}
\text{pointsec} : = \frac{k N_f}{t^k}. 
\ee 
For the sake of simplicity, we summarize the parameters and hyper-parameters for each case in Table \ref{tab:OverviewHP}.
\begin{table}[h!t]
\renewcommand\arraystretch{1.5}
\begin{center}
\footnotesize
\begin{tabular}{
    |>{\centering\arraybackslash}m{2.5cm}
    |>{\centering\arraybackslash}m{1.1cm}
    |>{\centering\arraybackslash}m{0.8cm}
    |>{\centering\arraybackslash}m{0.8cm}
    |>{\centering\arraybackslash}m{1.1cm}
    |>{\centering\arraybackslash}m{1.1cm}
    |>{\centering\arraybackslash}m{1.1cm}
    |>{\centering\arraybackslash}m{.8cm}|
    }
    \hline
\multirow{2}{*}{Case}  & learning rate & width & depth& \multirow{2}{*}{iterations} & \multirow{2}{*}{$|\Theta|$} & \multirow{2}{*}{$N^\textup{test}$}  & \multirow{2}{*}{$\sigma$} \\ 
  &$l_r$&   $W$ &$L-1$ &  & & &\\ \hline \hline
 1D Laplace& $10^{-4}$ & $50$ & $4$  &  $20000$ & $7801$ & $N_f$ & $\tanh$ \\ \hline 
 1D Schrödinger & $10^{-4}$ & $50$ & $4$ & $30000$ & $30802$ & $N_f$ & $\tanh$ \\\hline
 2D Navier-Stokes & $10^{-4}$ & $20$ & $8$ & $30000$ & $3604$ & $N_f$ & $\tanh$ \\ \hline 
 \end{tabular}
\end{center}
\caption{Overview of the parameters and hyper-parameters for each case.}
\label{tab:OverviewHP} 
\end{table}

\subsection*{1D Laplace}\label{1DLaplace}We first consider the 1D Laplace equation in $D = [-1,7]$ as being:
\be 
- \Delta u = f \quad \textup{in }\quad D \quad\textup{with}\quad f = \pi^2 \sin(\pi x) \quad \textup{and}\quad  u(-1)=u(7)=0.
\ee 
Acknowledge that $u(x) = \sin(\pi x)$. We solve the problem for:
$$
N_f = 2^i \quad  i= 3,\cdots ,1\rev{6} \quad \textup{and}\quad N_f= \rev{100,200,300,400}.
$$
We set $N_g=2$ and $N_u = 0$. Points in $\mT_D$ are generated randomly over $D$, and $\mT_b = \{-1,7\}$. The residual in Eq.~\eqref{eq:PDEresiduals}:
$$
\xi^f = -\Delta u - f
$$
yields the loss:
$$
\mL_\theta = \mL_\theta^f + \mL_\theta^b \quad \textup{with} \quad \mL_\theta^f := \frac{1}{N_f} \sum_{\mT_f}|\xi_f|^2 \quad \textup{and}\quad \mL_\theta^b := \frac{1}{2} \left( |u(-1)|^2 + |u(7)|^2\right).
$$
\subsubsection*{$h$-analysis}We perform the $h$-analysis for the error as portrayed before in Figure \ref{fig:transient}. The asymptotic regime occurs between $N_f=\rev{64}$ and $N_f=\rev{400}$, with a precision of $\rho=8$ in accordance with general results for $h$-analysis of traditional solvers. The permanent regime shows a slight improvement in accuracy, with \rev{mean} ``Error'' dropping from \rev{$8.75 \times 10^{-3}$ for $N_f=400$ to $3.91 \times 10^{-3}$ for $N=65536$}. To complete the $h$-analysis, Figure~\ref{fig:convergenceLaplace} shows the convergence results of ADAM optimizer for all the values of $N_f$. This plot reveals that each regime exhibits similar patterns. The high variability in convergence during the transition regime is particularly interesting, with some runs converging and others not. In the permanent regime, the convergence shows almost identical and stable patterns irrespective of $N_f$.
\begin{figure}[!htb]
\center
\includegraphics[width=\textwidth]{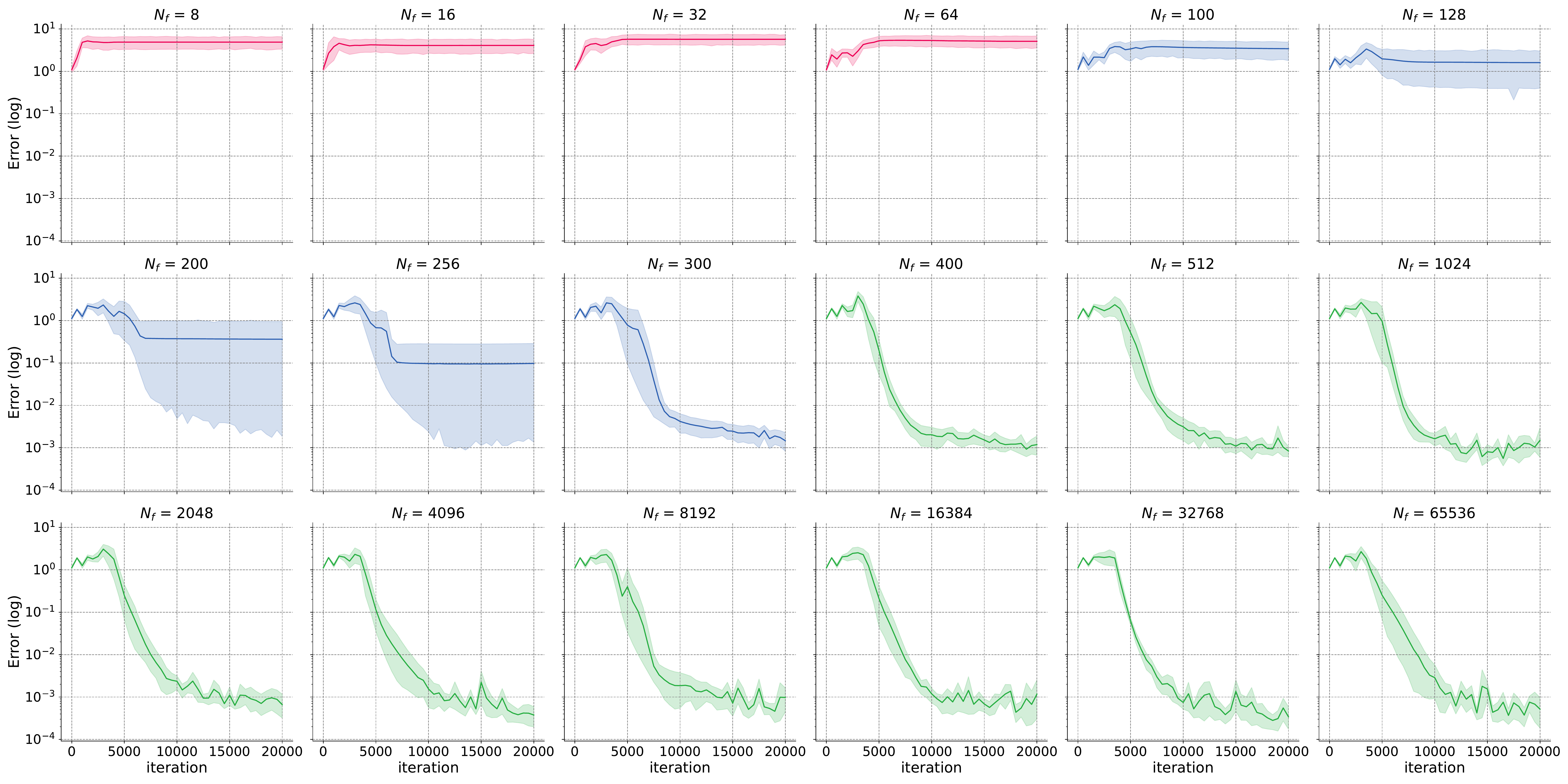}
\caption{1D Laplace: Convergence error for ADAM v/s $N_f$ for $l_r=10^{-4}$ and $20000$ iterations.}
\label{fig:convergenceLaplace}
\end{figure}

\rev{Furthermore, we plot the training and test losses in Figure~\ref{fig:lossesLaplace}. Acknowledge that the validation loss and ``Error'' show similar behaviors. We use this figure as a reference to define each transition regime. In particular, it hints that the permanent regime is reached for $N=400$ as the relative error at best iteration between $\mL_\theta^\text{train}$ and $\mL_\theta^\text{test}$ drops from $1.60 \times 10^{-1}$ to $6.02 \times 10^{-5}$. For the sake of precision, the value for $N=512$ is $2.20 \times 10^{-5}$.}

\begin{figure}[!htb]
\center
\includegraphics[width=\textwidth]{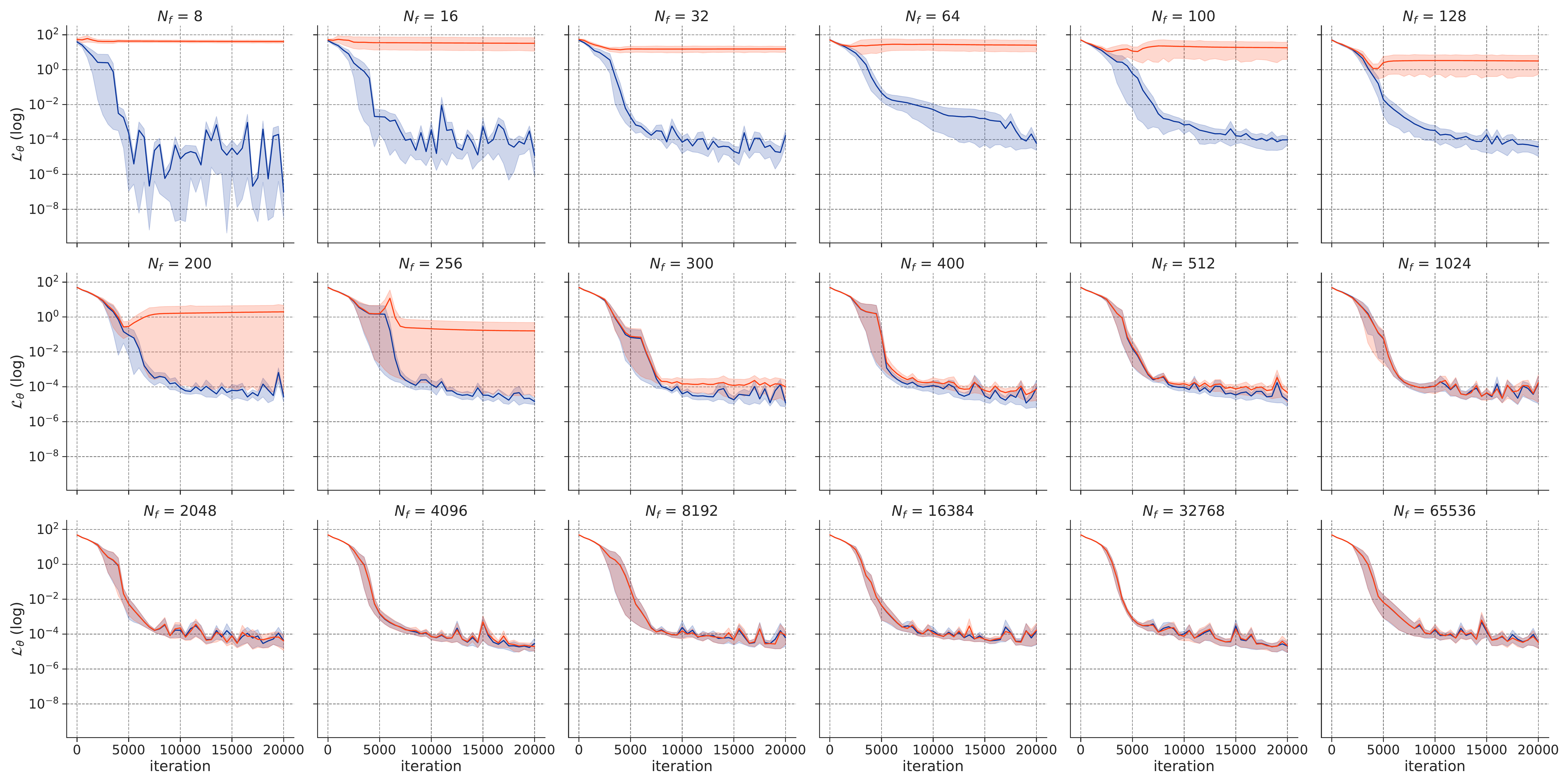}
\caption{1D Laplace: Training (blue) and test (orange) losses for ADAM v/s $N_f$ for $l_r=10^{-4}$ and $20000$ iterations.}
\label{fig:lossesLaplace}
\end{figure}

\subsubsection*{Data-parallel implementation}
We set $N_{f,1}\equiv N_1= \rev{64}$ and compare both weak and strong scaling to original implementation, referred to as ``no scaling''.
\begin{figure}[!htb]
\center
\includegraphics[width=.8\textwidth]{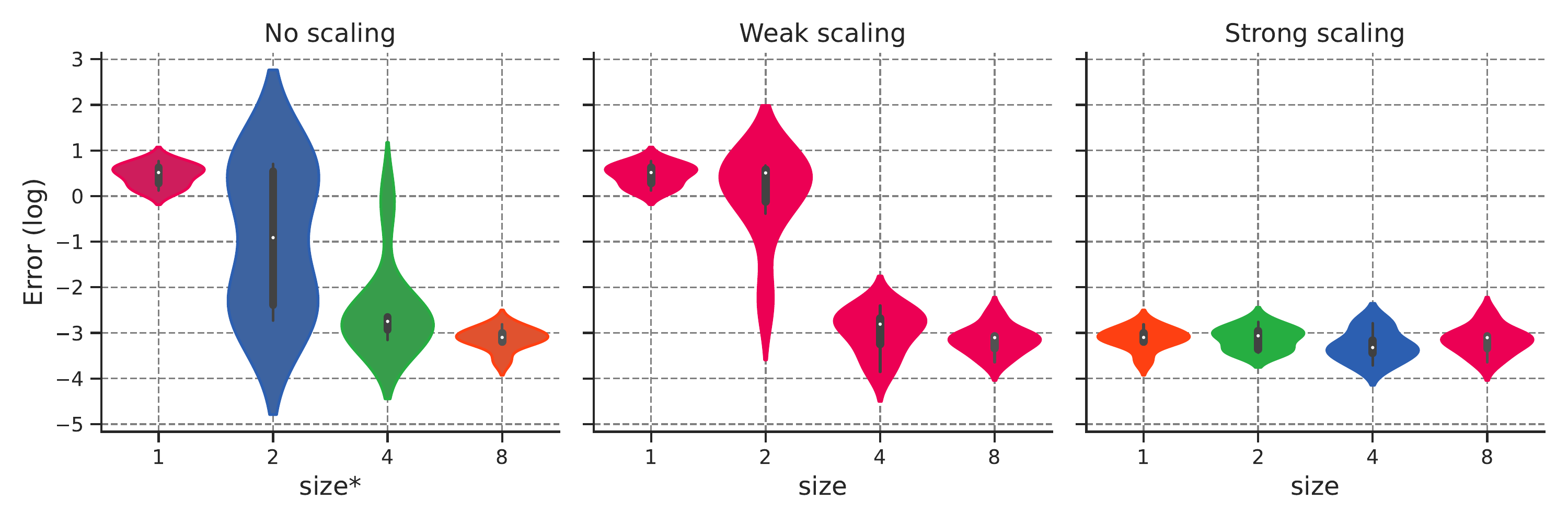}
\caption{Error v/s $\size^*$ for the different scaling options.}
\label{fig:LaplaceAccuracy}
\end{figure}

We provide a detailed description of Fig.~\ref{fig:LaplaceAccuracy}, as it will serve as the basis for future cases:
\begin{itemize}\setlength{\itemsep}{0pt}
    \item Left-hand side: Error for $\size^* \in\{1,2,4,8\}$ corresponding to $N_f\in \{64,128, \rev{256,512}\}$. 
    \item Middle: Error for weak scaling with $N_{1}=64$ and for $\size \in \{1,2,4,8\}$;
    \item Right-hand side: Error for strong scaling with $N_{1}=\rev{512}$ and for $\size \in \{1,2,4,8\}$
\end{itemize}
To reduce ambiguity, we use the $*$-superscript for no scaling, as $\size^*$ is performed over $1$ rank. The color for each violin box in the figure corresponds to the number of \rev{domain collocation points} used for each GPU. 

Fig.~\ref{fig:LaplaceAccuracy} demonstrates that both weak and strong scaling yield similar convergence results to their unaccelerated counterpart. This result is one of the main findings in this work: PINNs scale properly with respect to accuracy, validating the intuition behind $h$-analysis and justifying the data-parallel approach. This allows one to move from pre-asymptotic to permanent regime by using weak scaling, or leverage the cost of a permanent regime application by dispatching the training points over different workers. Furthermore, the hyper-parameters, including the learning rate, remained unchanged.

Next, we summarize the data-parallelization results in Table \ref{tab:OverviewPoisson1D} with respect to $\size$. 
\begin{table}[ht!]
\renewcommand\arraystretch{1.5}
\begin{center}
\footnotesize
\begin{tabular}{
    |>{\centering\arraybackslash}m{.8cm}
    ||>{\centering\arraybackslash}m{1.5cm}
    |>{\centering\arraybackslash}m{1.5cm}
    ||>{\centering\arraybackslash}m{1.5cm}
    |>{\centering\arraybackslash}m{1.5cm}
    |>{\centering\arraybackslash}m{1.5cm}
    |>{\centering\arraybackslash}m{1.5cm}
    }
    \hline
& \multicolumn{2}{c||}{$t^\text{500}$}& \multicolumn{2}{c|}{$E_\text{ff}$}\\ \hline 
$\size$ & Weak scaling & Strong scaling  & Weak scaling & Strong scaling\\ \hline\hline
$1$ & $\rev{4.08} \pm \rev{0.16}$ & $\rev{4.19} \pm \rev{0.18}$  & $-$ & $-$\\ \hline
$2$ & $ \rev{5.23}\pm \rev{0.13}$ & $ \rev{5.21}\pm \rev{0.16}$  & $\rev{78.01}\%$ &  $\rev{80.42}\%$\\ \hline
$4$ & $\rev{5.58}\pm \rev{0.20}$ & $\rev{5.65}\pm \rev{0.09}$ & $\rev{73.11}\%$ & $\rev{74.16}\%$\\ \hline
$8$ & $\rev{6.12}\pm \rev{0.23}$ &  $\rev{6.12}\pm \rev{0.23}$ & $\rev{66.67}\%$ &  $\rev{68.46}\%$\\ \hline
 \end{tabular}
\end{center}
\caption{1D Laplace: $t^{500}$ in seconds and efficiency $E_\text{ff}$ for the weak and strong scaling.}
 \label{tab:OverviewPoisson1D} 
\end{table}
In the first column, we present the time required to run $500$ iterations for ADAM, referred to as $t^\text{500}$. This value is averaged over one run of $30000$ iterations with a heat-up of $1500$ iteration (i.e.~we discard the values corresponding to iteration\rev{s} $0,500$ and $1000$). We present the resulting mean value $\pm$ standard deviation for the resulting vector. The second column, displays the efficiency of the run, evaluated with respect to $t^\text{500}$. 

\rev{Table \ref{tab:OverviewPoisson1D}} reveals that data-based acceleration results in supplementary training times as anticipated. Weak scaling efficiency varies between $\rev{78.01}\%$ for $\size=\rev{2}$ to $\rev{66.67}\%$ for $\size=8$, resulting in a speed-up of $\rev{5.47}$ when using $8$ GPUs. Similarly, strong scaling shows similar behavior. Furthermore, it can be observed that $\size=1$ yields almost equal $t^{500}$ for $N_f=\rev{64}$ ($\rev{4.08}s$) and $N_f=\rev{512}$ ($\rev{4.19}s$).

To conclude, the Laplace case is not reaching its full potential with Horovod due to the small batch size $N_{1}$. However, increasing the value of $N_{1}$ in \rev{next} cases will lead to a noticeable increase in efficiency.

\subsection*{1D Schrödinger}\label{Schrodinger}We solve the non-linear Schrödinger equation along with periodic BCs (refer to \cite[Section 3.1.1]{RAISSI2019686}) given over $\overline{D}\times \rev{[0, T ] }$ with $D:=(-5,5)$ and $T := \pi/ 2$:
\begin{align*}
\imath u_t + 0.5 u_{xx} + |u|^2 u &= 0 \quad \textup{in} \quad D \times (0,T),\\
u(-5,t) & =u(5,t), \rev{\quad t \in [0,T],}\\
\partial_x u(-5,t) & = \partial_x u(5,t),\rev{\quad t \in [0,T],}\\
u(x,0) & = 2 \sech(x),\rev{\quad x \in  D,}
\end{align*}
where $u(x,t) = u^0(x,t) + \imath u^1(x,t)$. We apply PINNs to the system with $m=2$ \rev{in Eq.~\eqref{eq:DNN}} as $(u^0_\theta,v^1_\theta) \in \IR^m=\IR^2$. We set $N_g=N_\hbar = 200$.
\subsubsection*{$h$-analysis}

To begin with, we perform the $h$-analysis for the parameters in Fig.~\ref{fig:transientSchrodinger}.
\begin{figure}[!htb]
\center
\includegraphics[width=\textwidth]{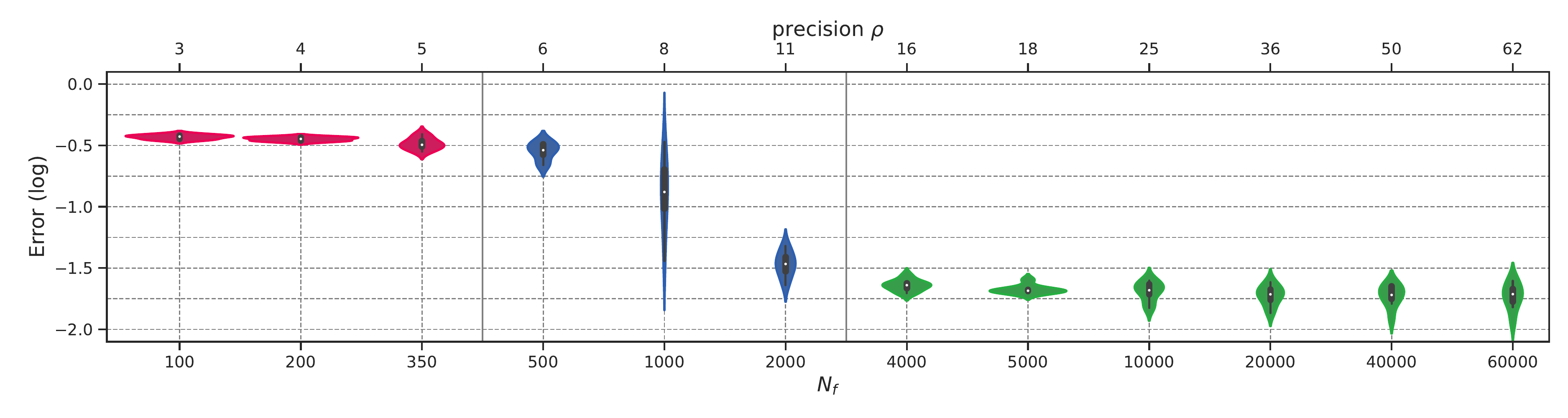}
\caption{Schrödinger: Error v/s $N_f$.}
\label{fig:transientSchrodinger}
\end{figure}
Again, the transition regime for the total execution time distribution begins at a density of $\rho=5$ and spans between $N_f=350$ and $N_f=4000$. At higher magnitudes of $N_f$, the error remains approximately the same. To illustrate this more complex case, we present the total execution time distribution in Fig.~\ref{fig:transientTimesSchrodinger}. 
\begin{figure}[!htb]
\center
\includegraphics[width=\textwidth]{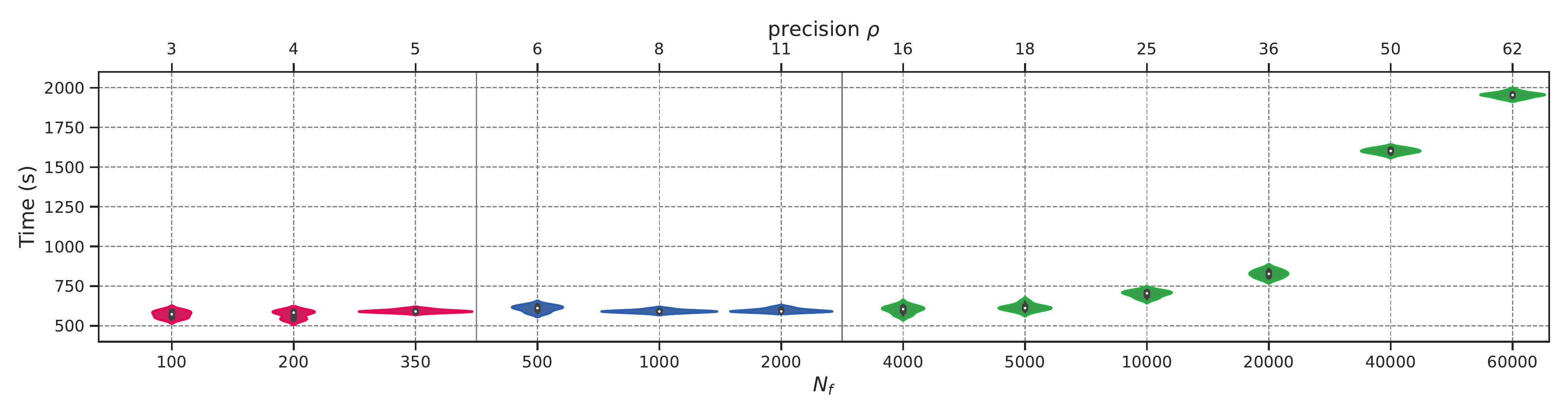}
\caption{Schrödinger: Time (in seconds) v/s $N_f$.}
\label{fig:transientTimesSchrodinger}
\end{figure}
We note that the training times remain stable for $N_f\leq4000$. This observation is of great importance and should be emphasized, as it adds additional parameters to the analysis. Our work primarily focuses on error in $h$-analysis, however, execution time and memory requirements are also important considerations. We see that in this case, weak scaling is not necessary (the optimal option is to use $\size = 1$ and $N_f = 4000$). Alternatively, strong scaling can be done with $N_{f,1} = 4000$. 

To gain further insight into the variability of the transition regime, we focus on $N_f = 1000$. We compare the solution for $\seed = 1234$ and $\seed = 1236$ in Fig.~\ref{fig:SchrodingerSeeds}.
\begin{figure}[!htb]
\centering
\begin{subfigure}{.5\textwidth}
  \centering
  \includegraphics[width=\linewidth]{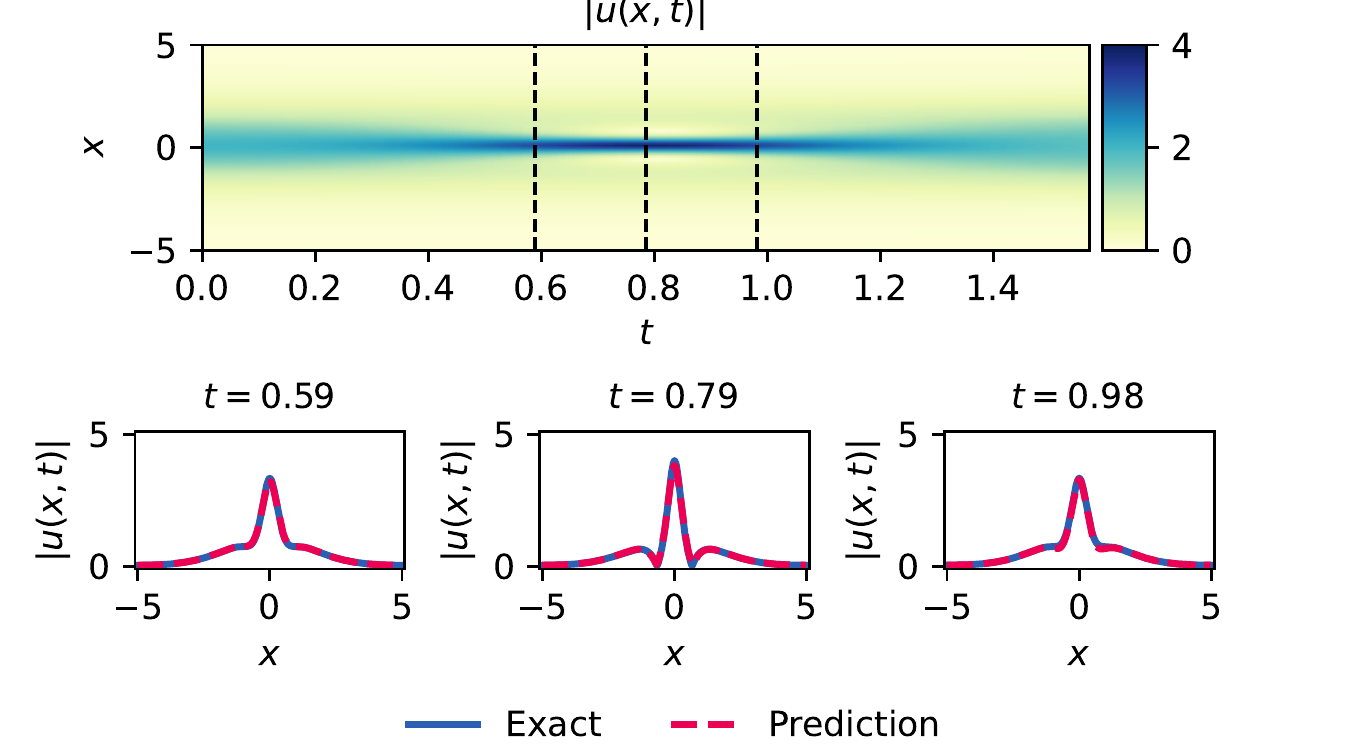}
  \caption{Results for $\seed=1234$}
  \label{fig:sub1}
\end{subfigure}%
\begin{subfigure}{.5\textwidth}
  \centering
  \includegraphics[width=\linewidth]{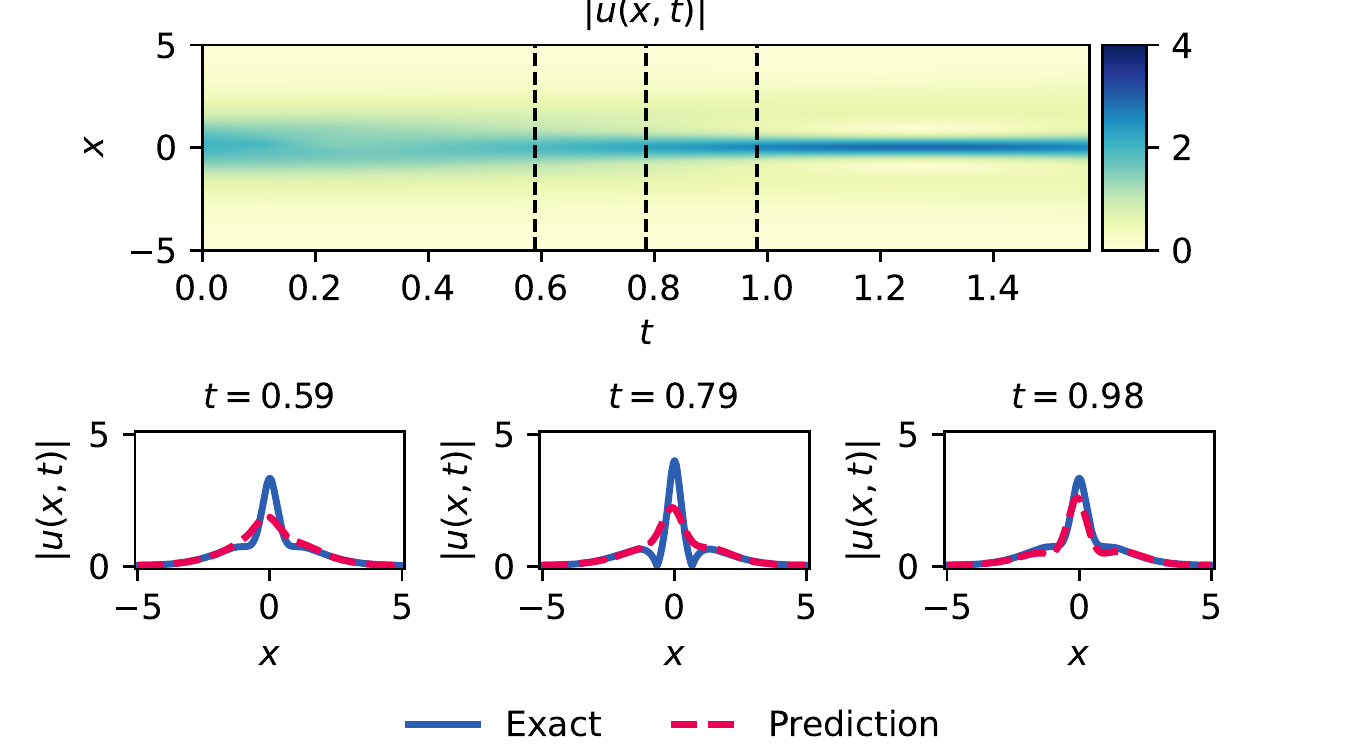}
  \caption{Results for $\seed=1236$}    
  \label{fig:sub2}
\end{subfigure}
\caption{Schrödinger: Solution for $N_f=1000$ for two differents seeds.}
\label{fig:SchrodingerSeeds}
\end{figure}
The upper figures depict $|u(t,x)|$ predicted by the PINN. The lower figures show a comparison between the exact and predicted solutions are plotted for $t\in \{0.59,0.79,-0.98\}$. It is evident that the solution for $\seed=1234$ closely resembles the exact solution, whereas the solution for $\seed=1236$ fails to accurately represent the solution near $x=0$, thereby illustrating the importance of achieving a permanent regime. \rev{Next, we show the training and testing losses in Fig.~\ref{fig:lossesSchrodinger}. We remark that the training and test losses converge for $N_f >200$. Analysis of the train-test gap showed that it converged as $\mO(N_f^{-1})$. Visually, one can assume that the losses are close enough for $N_f=4000$ (or $N_f=8000$), in accordance with the $h$-analysis performed in Fig.~\ref{fig:transientSchrodinger}.}

\begin{figure}[!htb]
\center
\includegraphics[width=\textwidth]{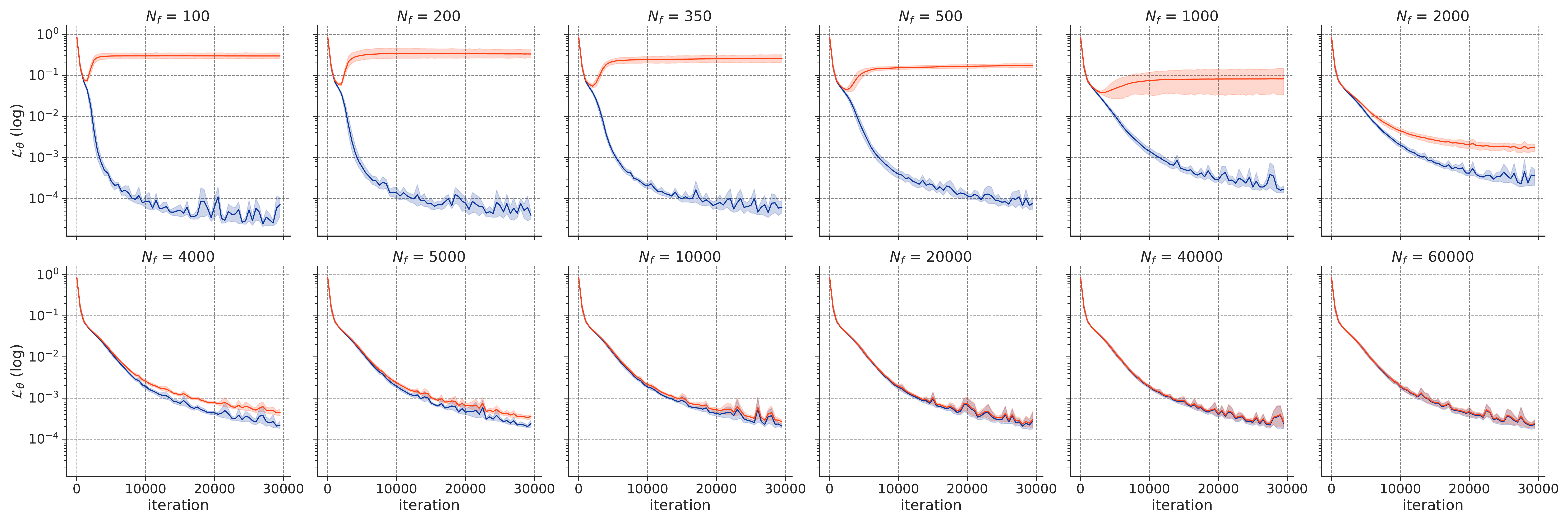}
\caption{Schrödinger: Training (blue) and test (orange) losses for ADAM v/s $N_f$ for $l_r=10^{-4}$ and $20000$ iterations.}
\label{fig:lossesSchrodinger}
\end{figure}

\subsubsection*{Data-parallel implementation}
We compare the error for unaccelerated and data-parallel implementations of simulations for $N_{f,1}\equiv N_{1}=500$ in Fig.~\ref{fig:SchrodingerAccuracy}, analogous to the analysis in Fig.~\ref{fig:LaplaceAccuracy}.
\begin{figure}[!htb]
\center
\includegraphics[width=.8\textwidth]{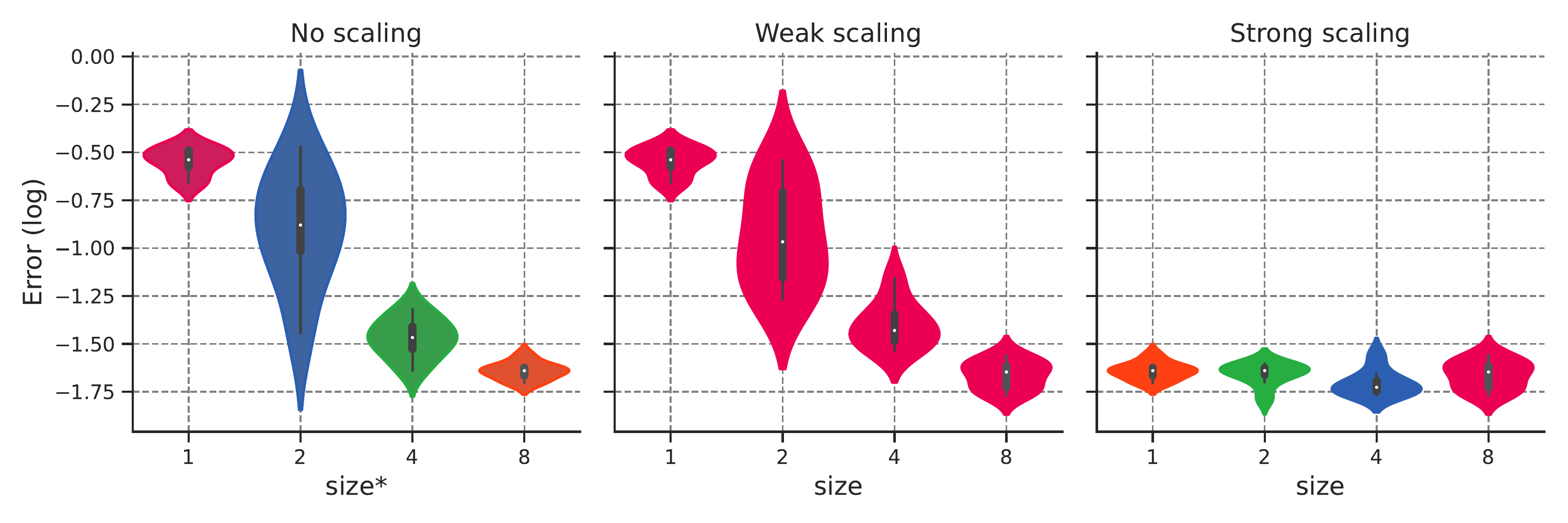}
\caption{Schrödinger: Error v/s $\size^*$ for the different scaling options.}
\label{fig:SchrodingerAccuracy}
\end{figure}
Again, the error is stable with $\size$. Both no scaling and weak scaling are similar. Strong scaling is unaffected by $\size$. We plot the training time in Fig.~\ref{fig:SchrodingerTime}.
\begin{figure}[!htb]
\center
\includegraphics[width=.8\textwidth]{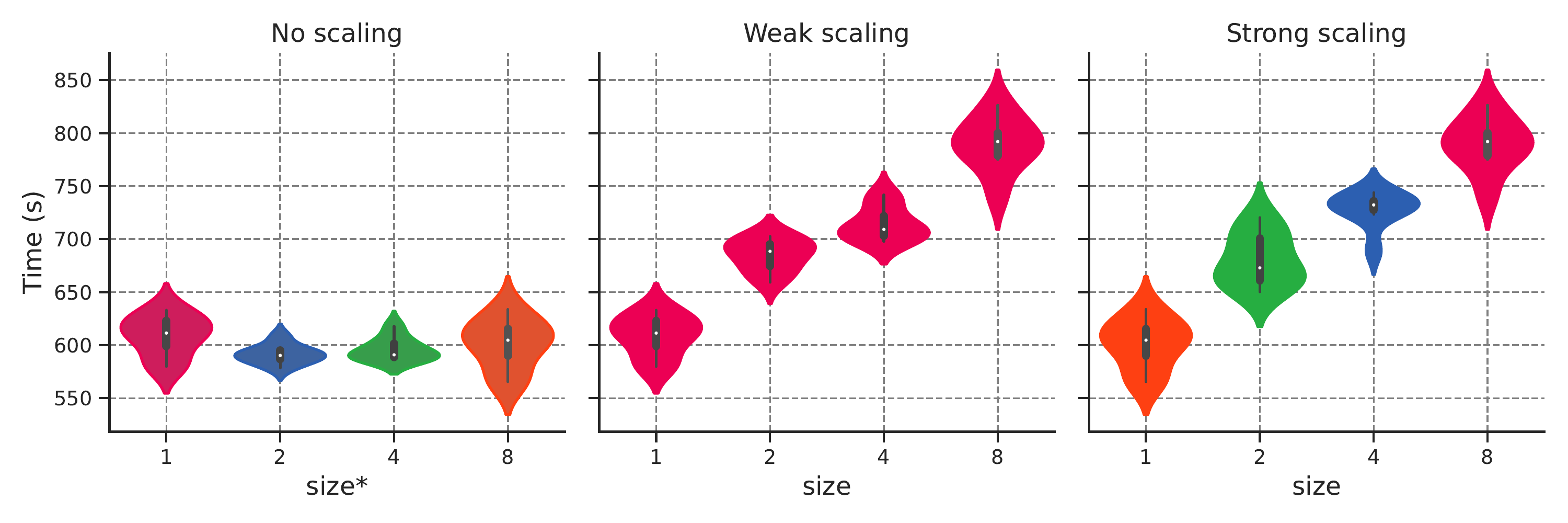}
\caption{Schrödinger: Time (in seconds) v/s $\size^*$ for the different scaling options.}
\label{fig:SchrodingerTime}
\end{figure}
We observe that both weak and strong scaling increase linearly and slightly with $\size$. Both scaling show similar behaviors. Fig.~\ref{fig:SchrodingerEfficiency} \rev{portrays the number of training points processed per second (refer to Eq.~\eqref{eq:pointsec})} and the efficiency with respect to $\size$, with white bars representing the ideal scaling. Efficiency $E_\text{ff}$ shows a gradual decrease with $\size$, with results surpassing those of the previous section. The efficiency for $\size=8$ reaches $77\rev{.22}\%$ and $76\rev{.20}\%$ respectively for weak and strong scaling, representing a speed-up of $6.\rev{18}$ and $6.1\rev{0}$.
\begin{figure}[!htb]
\center
\includegraphics[width=.8\textwidth]{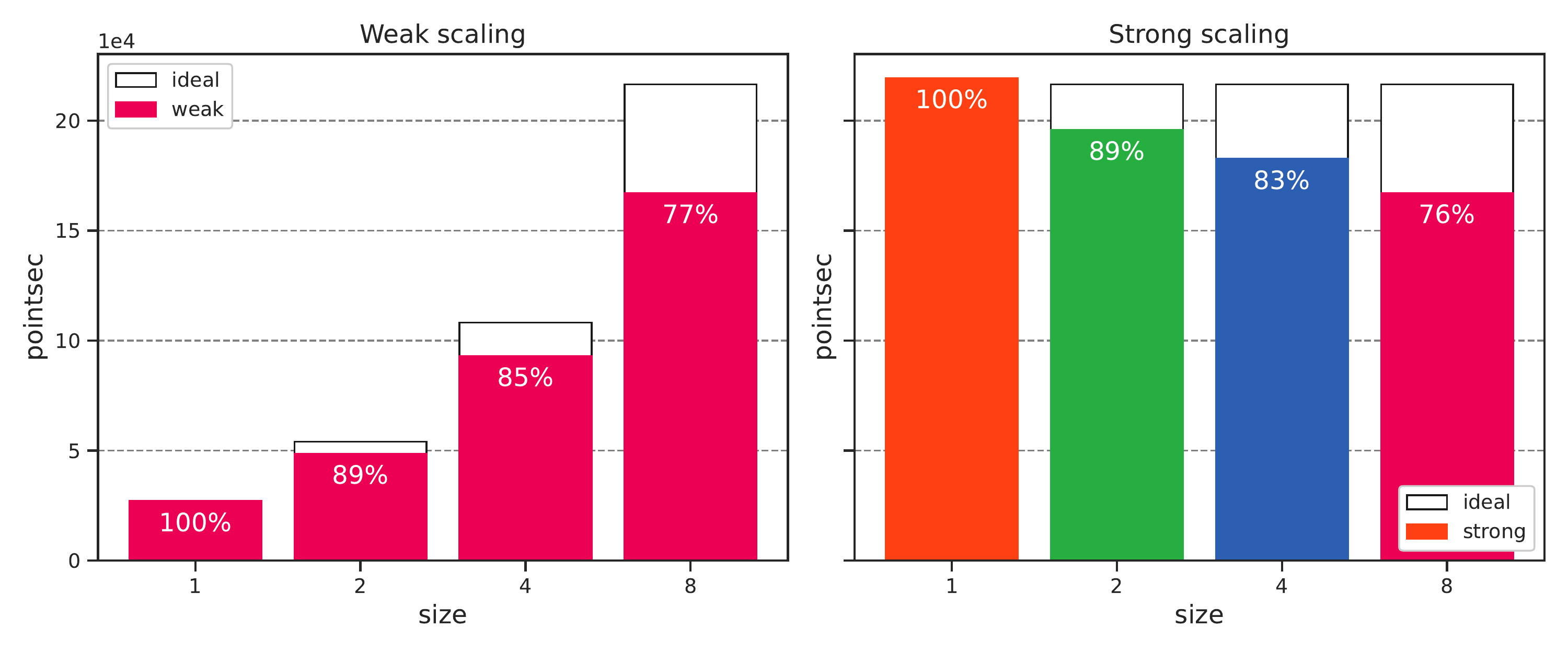}
\caption{Schrödinger: Efficiency v/s $\size$.}
\label{fig:SchrodingerEfficiency}
\end{figure}

\subsection*{Inverse problem: Navier-Stokes equation}\label{NavierStokes}We consider the Navier-Stokes problem with $D : =[-1,8] \times [-2,2]$, $T=20$ and unknown parameters $\uplambda_1\uplambda_2 \in \IR$. The resulting divergence-free Navier Stokes is expressed as follows:
\be\label{eq:NS2D}
\begin{cases}
u_t + \lambda_1 ( u u_x +  v u_y )& = - p_x + \lambda_2 (u_{xx} + u_{yy}) ,\\
v_t + \lambda_1 ( u v_x +  v v_y )& = - p_y + \lambda_2 (v_{xx} + v_{yy}) ,\\
u_x + u_t & = 0,
\end{cases}
\ee
wherein $u(x,y,t)$ and $v(x,y,t)$ are the $x$ and $y$ components of the velocity field, and $p(x,y,t)$ the pressure. We assume that there exists $\varphi(x,y,t)$ such that:
\be\label{eq:potential}
u = \varphi_x ,\quad \textup{and}\quad v = - \varphi_x.
\ee
Under Eq.~\eqref{eq:potential}, last row in Eq.~\eqref{eq:NS2D} is satisfied. The latter leads to the definition of residuals:
\begin{align*} 
\xi_1&: =  u_t + \lambda_1 ( u u_x +  v u_y ) +  p_x - \lambda_2 (u_{xx} + u_{yy})\\
\xi_2  & : = v_t + \lambda_1 ( u v_x +  v v_y ) + p_y -\lambda_2 (v_{xx} + v_{yy}).
\end{align*}
\rev{We introduce $N_f$ pseudo-random points $\bx^i \in D$ and observations $(u^i,v^i)= (u^i (\bx^i), v^i (\bx^i)) $,} yielding the loss:
\begin{align*}
\mL_\theta : = \mL_\theta^u  +  \mL^f_{\theta} 
\end{align*}
with
$$
\mL_\theta^u = \frac{1}{\rev{N_f}} \sum_{i=1}^{\rev{N_f}} \left(|u(\bx^i)-u^i|^2 + |v(\bx^i)-v^i|^2\right)\quad\textup{and}\quad \mL_\theta^f = \frac{1}{\rev{N_f}} \sum_{i=1}^{\rev{N_f}} \left(\xi_1(\bx^i)\rev{^2} + \xi_2(\bx^i)\rev{^2}\right).
$$
Throughout this case, we have $N_f=M$, and we plot the results with respect to $N_f$. Acknowledge that $N = 2 N_f$, and that both $\lambda_1,\lambda_2$ and the BCs are unkwown.
\subsubsection*{$h$-analysis}We conduct the $h$-analysis and show the error in Fig.~\ref{fig:transientNavierStokes}, showing no differences with previous cases.
\begin{figure}[!htb]
\center
\includegraphics[width=\textwidth]{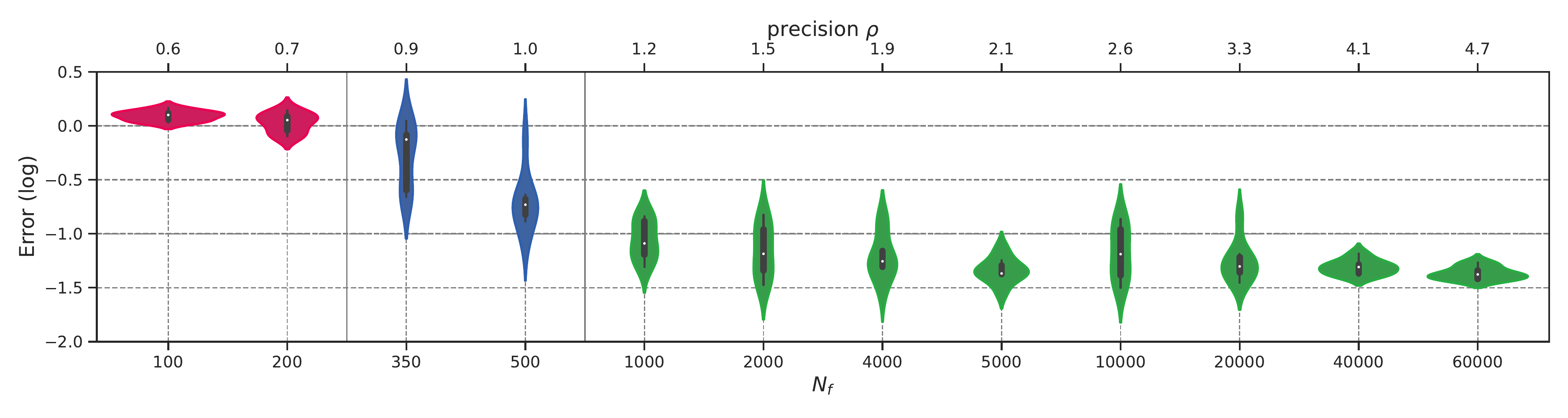}
\caption{Navier-Stokes: Error v/s $N_f$.}
\label{fig:transientNavierStokes}
\end{figure}
 Surprisingly, the permanent regime is reached only for $N_f=1000$, despite the problem being a $3$-dimensional, non-linear, and inverse one. This corresponds to low values of $\rho$, indicating that PINNs seem to prevent the curse of dimensionality. In fact, the it was achieved with only $1.2\rev{1}$ points per unit per dimension. The total training time is presented in Fig.~\ref{fig:transientTimesNavierStokes}, where it can be seen to remain stable up to $N_f=5000$, and then increases linearly with $N_f$.
\begin{figure}[!htb]
\center
\includegraphics[width=\textwidth]{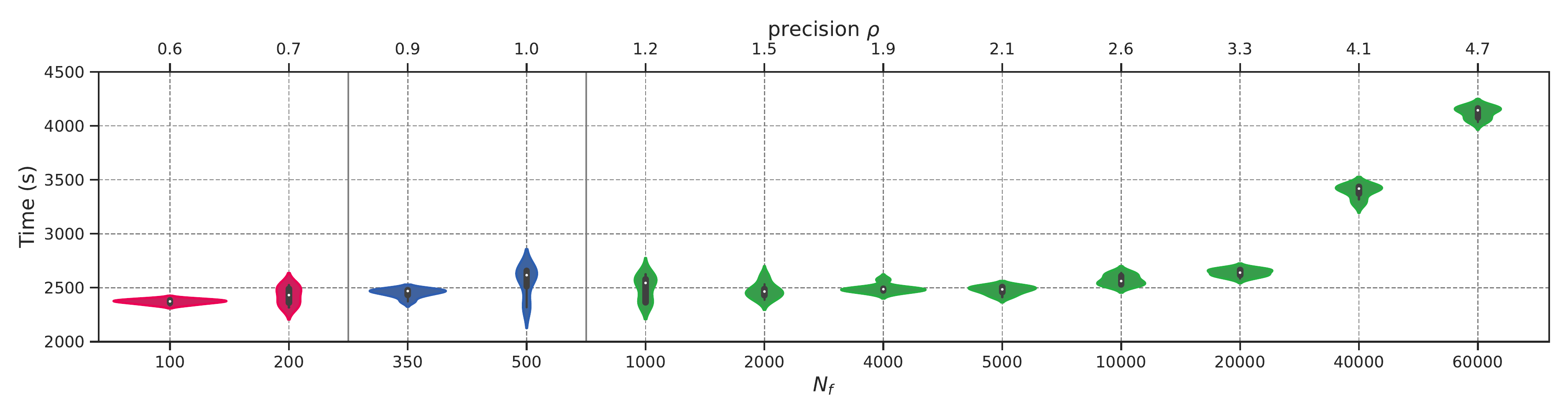}
\caption{Navier-Stokes: Time (in seconds) v/s $N_f$.}
\label{fig:transientTimesNavierStokes}
\end{figure}

\rev{Again, we represent the training and test losses in Fig.~\ref{fig:lossesNavier}. The train-test gap is shown to decrease for $N_f \geq 350$ as $\mO(N_f^{-1})$. Furthermore, the train-test gap can be considered as being small enough, visually, for $N_f = 1000$.}
\begin{figure}[!htb]
\center
\includegraphics[width=\textwidth]{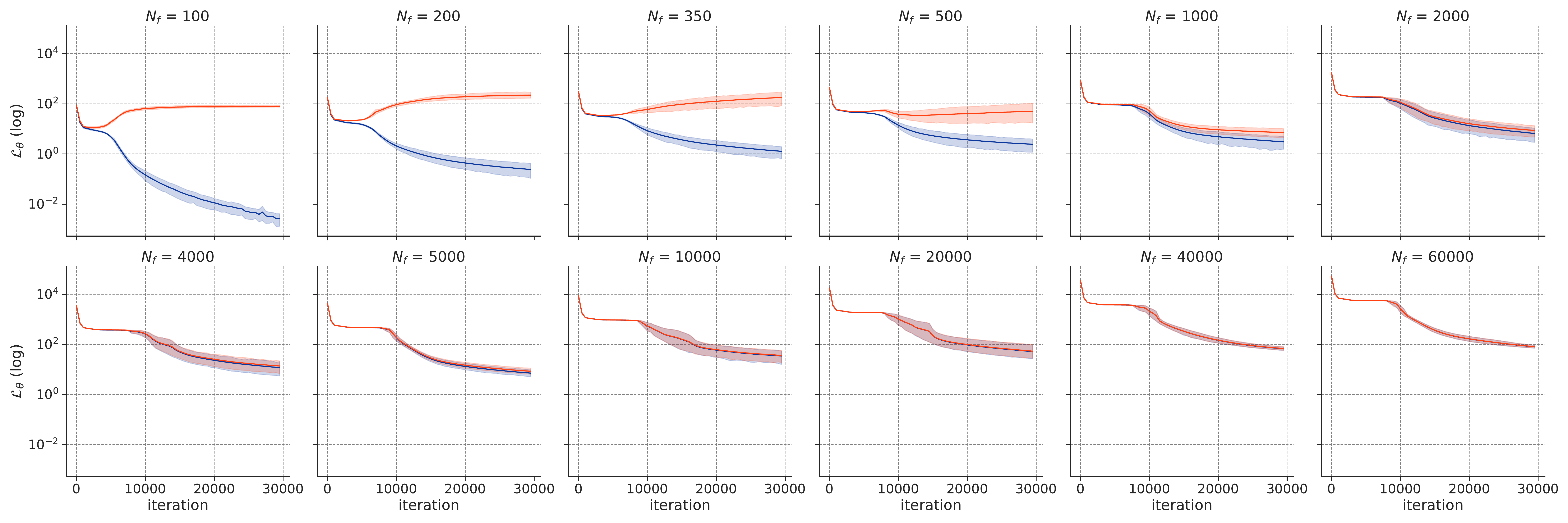}
\caption{Navier-Stokes: Training (blue) and test (orange) losses for ADAM v/s $N_f$.}
\label{fig:lossesNavier}
\end{figure}

\subsubsection*{Data-parallel implementation}We run the data-parallel simulations, setting $N_{f,1} \equiv N_{1} = M_1= 500$. As shown in Fig.~\ref{fig:NavierStokesAccuracy}, the simulations exhibit stable accuracy with $\size$.
\begin{figure}[!htb]
\center
\includegraphics[width=.8\textwidth]{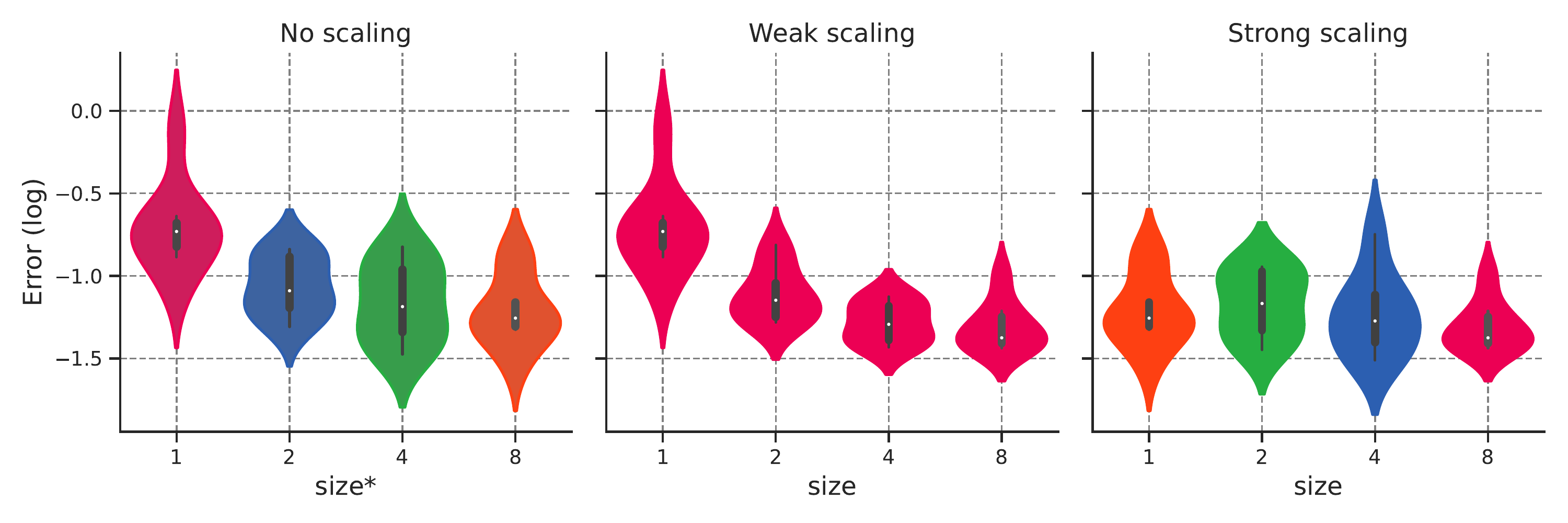}
\caption{Navier-Stokes: Error v/s $\size^*$ for the different scaling options.}
\label{fig:NavierStokesAccuracy}
\end{figure}
The execution time increases moderately with $\size$, as illustrated in Fig.~\ref{fig:NavierStokesTime}. The training time decreases with $N$ for the no scaling case. However, this behavior is temporary (refer to Fig.~\ref{fig:transientNavierStokes} before).
\begin{figure}[!htb]
\center
\includegraphics[width=.8\textwidth]{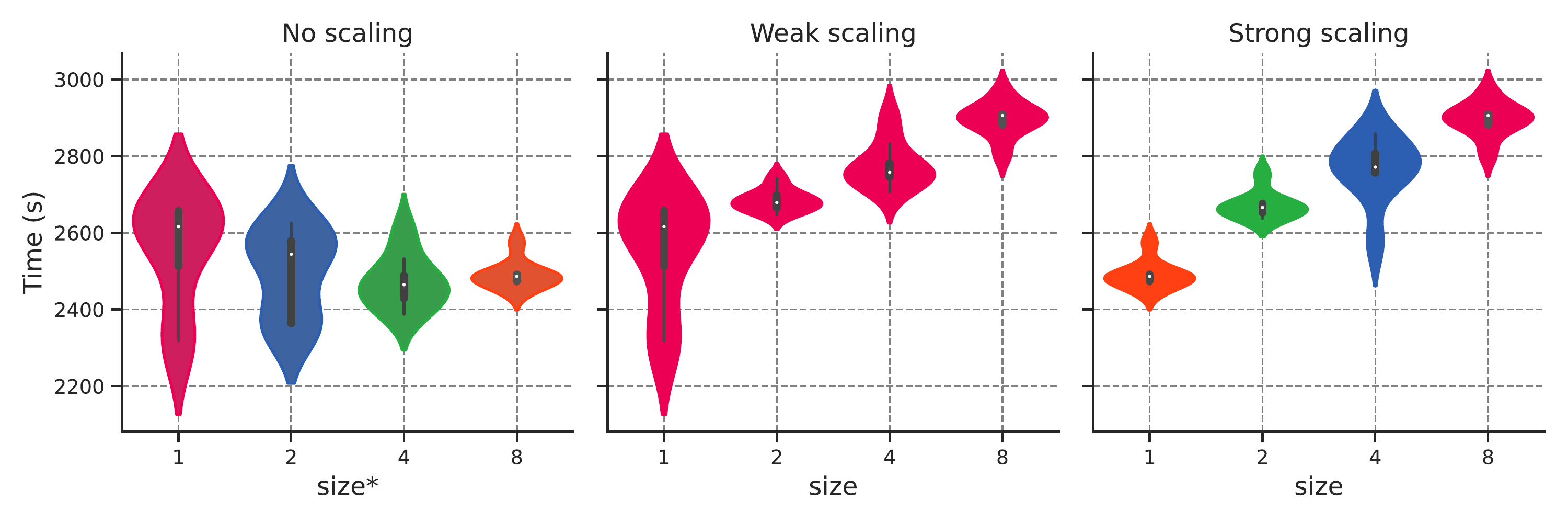}
\caption{Navier-Stokes: Time (in seconds) v/s $\size^*$ for the different scaling options.}
\label{fig:NavierStokesTime}
\end{figure}
We conclude our analysis by plotting the efficiency with respect to $\size$ in Fig.~\ref{fig:NavierStokesEfficiency}. Efficiency lowers with increasing $\size$, but shows the best results so far, \rev{with $80\rev{.55}\%$ (resp.~$86.31\%$) weak (resp.~strong) scaling efficiency for $\size = 8$}. \rev{For the sake of completeness, the weak efficiency for $N_{f,1}=50000$ and $\size=8$ improved to $86.15\%$.} This encouraging result sets the stage for further exploration of more intricate applications.

\begin{figure}[!htb]
\center
\includegraphics[width=.8\textwidth]{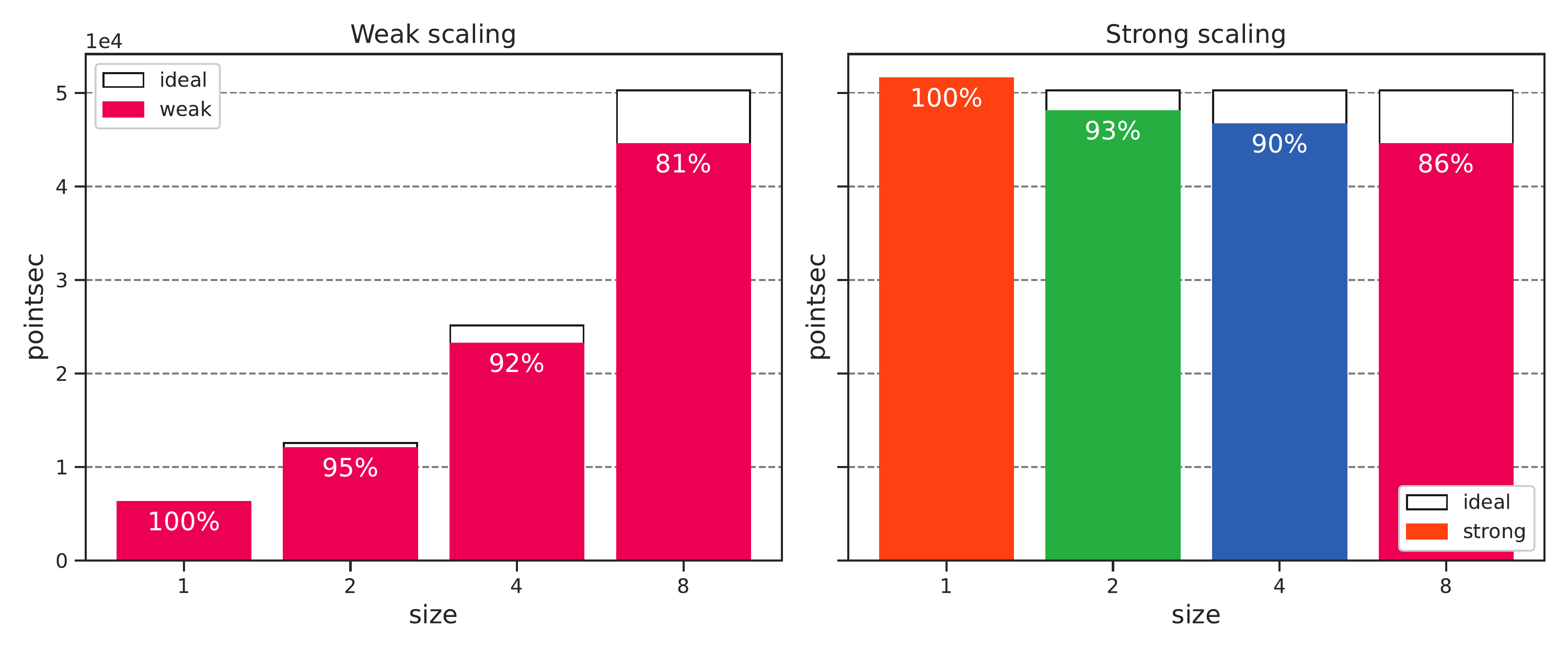}
\caption{Navier-Stokes: Efficiency v/s $\size$.}
\label{fig:NavierStokesEfficiency}
\end{figure}

\section*{Conclusion}\label{sec:Conclusion}
In this work, we proposed a novel data-parallelization approach for PIML with a focus on PINNs. We provided a thorough $h$-analysis and associated theoretical results to support our approach, as well as implementation considerations to facilitate implementation with Horovod data acceleration. Additionally, we ran reproducible numerical experiments to demonstrate the scalability of our approach. Further work include the implementation of Horovod acceleration to \href{https://github.com/lululxvi/deepxde/}{DeepXDE}\cite{lu2021deepxde} library, coupling of localized PINNs with domain decomposition methods, and application on larger GPU servers (e.g., with more than 100 GPUs).

\section*{Data availability}
The code required to reproduce these findings are available to download from
 \href{https://github.com/pescap/HorovodPINNs}{https://github.com/pescap/HorovodPINNs}.

\bibliography{references}


\section*{Acknowledgements}

The authors would like to thank the Data Observatory Foundation, ANID FONDECYT 3230088, FES-UAI postdoc grant, ANID PIA/BASAL FB0002, and ANID/PIA/ANILLOS ACT210096, for financially supporting this research.

\section*{Author contributions statement}P.E.I. conceived the experiment(s). All authors analyzed the results and reviewed the manuscript. 

\end{document}